\documentclass[aps,pra,twocolumn,superscriptaddress,notitlepage,noshowkeys,10pt]{revtex4-2}

\usepackage{blindtext,graphicx,amsmath,amsthm,physics,mathtools,float,comment,xcolor,amssymb,bbm,stackengine,caption,subcaption,float,adjustbox}

\newcommand{\rmd}{\mathrm{d}}
\newcommand{\rmi}{\mathrm{i}}
\newcommand{\rme}{\mathrm{e}}
\newcommand{\rmBloch}{\mathrm{Bloch}}

\newcommand{\dist}{\operatorname{dist}}
\newcommand{\Dom}{\operatorname{Dom}}

\stackMath                 

\newcommand\tsup[2][2]{%
 \def\useanchorwidth{T}%
  \ifnum#1>1%
    \stackon[-.5pt]{\tsup[\numexpr#1-1\relax]{#2}}{\scriptscriptstyle\sim}%
  \else%
    \stackon[.5pt]{#2}{\scriptscriptstyle\sim}%
  \fi%
}

\newtheorem{thm}{Theorem}
\newtheorem{proposition}[thm]{Proposition}
\newtheorem{lemma}[thm]{Lemma}

\newcommand{\hilb}{\mathcal{H}}
\newcommand{\ii}{\mathrm{i}}
\newcommand{\e}{\mathrm{e}}

\usepackage[unicode=true,
bookmarks=true,bookmarksopen=false,
breaklinks=false,pdfborder={0 0 0},colorlinks=true]
{hyperref}
\usepackage{xcolor}
\definecolor{cblue}{rgb}{0.16, 0.32, 0.75}
\definecolor{cred}{rgb}{0.7, 0.11, 0.11}
\hypersetup{%
	,linkcolor=cred
	,citecolor=cblue
	,urlcolor=black
}

\begin{document}
\title{
Robust quantification of spectral transitions in perturbed quantum systems
}
\date{\today}
\author{Zsolt Szab\'o}
\email{zsolt.szabo@students.mq.edu.au}
\affiliation{School of Mathematical and Physical Sciences, Macquarie University, NSW 2109, Australia}
\affiliation{ARC Centre of Excellence for Engineered Quantum Systems, Macquarie University, NSW 2109, Australia}
\author{Stefan Gehr}
\affiliation{Department Physik, Friedrich-Alexander-Universit\"at Erlangen-N\"urnberg, Staudtstra\ss e 7, 91058 Erlangen, Germany}
\author{Paolo Facchi}
\affiliation{Dipartimento di Fisica, Universit\`a di Bari, I-70126 Bari, Italy}
\affiliation{INFN, Sezione di Bari, I-70126 Bari, Italy}
\author{Kazuya Yuasa}
\affiliation{Department of Physics, Waseda University, Tokyo 169-8555, Japan}
\author{Daniel Burgarth}
\affiliation{Department Physik, Friedrich-Alexander-Universit\"at Erlangen-N\"urnberg, Staudtstra\ss e 7, 91058 Erlangen, Germany}
\author{Davide Lonigro}
\affiliation{Department Physik, Friedrich-Alexander-Universit\"at Erlangen-N\"urnberg, Staudtstra\ss e 7, 91058 Erlangen, Germany}

\begin{abstract}
A quantum system subject to an external perturbation can experience leakage between uncoupled regions of its energy spectrum separated by a gap. To quantify this phenomenon, we present two complementary results. First, we establish time-independent bounds on the distances between the true dynamics and the dynamics generated by block-diagonal effective evolutions constructed via the Schrieffer--Wolff and Bloch methods. Second, we prove that, under the right conditions, this leakage remains small \textit{eternally}. That is, we derive a time-independent bound on the leakage itself, expressed in terms of the spectral gap of the unperturbed Hamiltonian and the norm of the perturbation, ensuring its validity for arbitrarily large times. Our approach only requires a finite spectral gap, thus accommodating continuous and unbounded spectra. Finally, we apply our bounds to specific systems of practical interest.  
\end{abstract}

\maketitle

\section{Introduction}\label{ch.intro}
Quantum systems governed by gapped Hamiltonians have a natural protection against transitions out of low-energy subspaces when perturbed. This property is important in areas such as adiabatic quantum computing~\cite{Farhi2000,albash_adiabatic_2018}, quantum error correction~\cite{kitaev_fault-tolerant_2003,chesi_thermodynamic_2010}, and condensed matter physics, where controlling or bounding undesired population transfer out of a given eigenspace is key to maintaining coherent dynamics. In practical scenarios, one needs rigorous and quantitative estimates to guarantee that small perturbations applied to the system do not induce large transitions at any time. In particular, it is desirable to obtain \textit{eternal} estimates, that is, valid for arbitrarily large times.

One technique for investigating these bounded transitions relies on effective Hamiltonian theory, developed to capture the essential low-energy physics of a complex system by integrating out high-energy degrees of freedom, yielding a simplified yet accurate description of its relevant dynamics. The Schrieffer--Wolff transformation, named after the authors of Ref.~\cite{schrieffer_relation_1966}, accomplishes subspace decoupling through a unitary transformation that is typically determined recursively. A self-consistent review of the method with examples was authored by Bravyi \textit{et~al.}~\cite{bravyi_schriefferwolff_2011}. However, the existing literature does not generally address the error involved in such approximations---namely, how closely the unitary dynamics generated by the effective Hamiltonian aligns with the true system evolution in a general context. One goal of this manuscript is to remedy this shortcoming.

\begin{figure*}[] 
     \centering
     \includegraphics[width=0.9\textwidth]{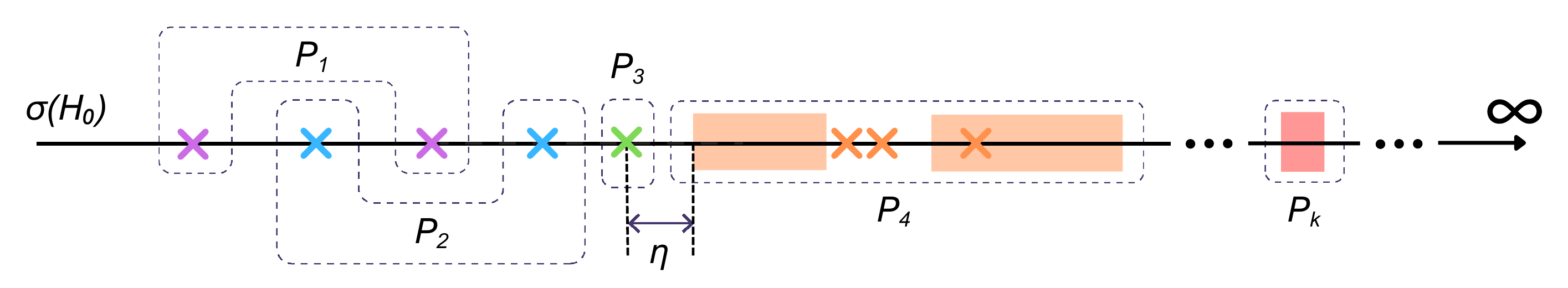}
     \caption{Example of the coarse-graining of a Hamiltonian with generic spectrum. The spectral gap $\eta$ is the minimal distance between the coarse-grained spectral components. }
     \label{fig:theFigure}
\end{figure*}

Another effective method, originally introduced by Claude Bloch~\cite{bloch_sur_1958}, relies on finding an effective Hamiltonian that, in the limit in which the relative strength $\gamma$ of the drift Hamiltonian with respect to the perturbation becomes large, does not provoke transitions between distinct components of the spectrum of the unperturbed system and, at the same time, reproduces the dynamics of the actual system. This can be obtained by solving a system of operator equations---the Bloch equations. 
Lindgren~\cite{lindgren_rayleigh-schrodinger_1974}, Durand~\cite{durand_direct_1983}, and later Jolicard~\cite{jolicard_effective_1987} described recursive approaches to solving the Bloch equations, while Killingbeck and Jolicard~\cite{killingbeck_bloch_2003} offered a detailed summary on how a suitably defined Bloch wave operator block-diagonalizes the full Hamiltonian with respect to the unperturbed drift. Their discussion, however, does not allow one to quantify how off-diagonal leakage terms scale with the perturbation strength for large times, nor do they look at systems with unbounded Hamiltonians.

More recently, Burgarth \textit{et~al.}~\cite{eternal} investigated the phenomenon of eternal adiabaticity, showing how one can derive time-independent bounds on the leakage for autonomous generators with a finite spectral gap. Similarly, Facchi \textit{et~al.}~\cite{facchi_robustness_2024} discussed the large-time stability of quantum symmetries.
In both cases, these results are obtained for systems with discrete energy spectra and are formulated in a \textit{fine-grained} setup, that is, each eigenspace is considered separately, and the spectral gap $\eta$ characterizes well-defined energy-level separations between individual eigenvalues.

On the other hand, in applications, one is often interested in the \textit{coarse-grained} leakage between distinct components of the spectrum---each component possibly consisting of multiple eigenvalues, or even containing continuous spectrum. This is the case, for instance, in solid-state physics, where one is usually led to consider quantum systems whose spectra consist of continuous energy bands, and only transitions between energies pertaining to distinct bands are of interest. The aforementioned results do not cover this scenario.

This work will address this problem and provide an eternal upper bound to the leakage between distinct components of the unperturbed energy spectrum. We will work in the infinite-dimensional setting, which will allow us to handle possibly continuous spectra without the need for any limiting procedure. Our results are robust, in the sense that they are independent of the specifics of the unperturbed energy spectrum of the system: only the value of the coarse-grained spectral gap enters the bound, which is otherwise independent of the number of such components as well as the dimension of the vector subspace corresponding to each component of the spectrum (which might as well be infinite). Lastly, unbounded energy spectra are also covered, as long as the norm of the perturbation stays finite.
To achieve this result, we will first adapt the Bloch wave operator and the Schrieffer--Wolff transformation methods to a more general coarse-grained spectral decomposition, and obtain a single, uniform upper bound on the difference between the effective and actual dynamics.

The remainder of this paper is organized as follows. In Sec.~\ref{sec:setting}, we gather the necessary notation and definitions. In Sec.~\ref{sec:bloch}, we describe the Bloch method for coarse-grained spectral decomposition, and apply it to derive an explicit expression for the Bloch generator of an effective dynamics which preserves the block-diagonal structure of the unperturbed spectrum. In Sec.~\ref{sec:sw}, we find a Hermitian effective generator by means of a Schrieffer--Wolff transformation. In Sec.~\ref{sec:unbounded}, we use these results to find an eternal bound on the leakage which is also rigorously shown to also hold in cases where the unperturbed Hamiltonian has an unbounded energy spectrum. Lastly, in Sec.~\ref{sec:examples}, we apply our bound to specific examples of practical interest, and in Sec.~\ref{sec:conclusion}, we provide concluding remarks and highlight potential further extensions to our results.

\section{Setting and Notation}\label{sec:setting}
We shall consider a quantum system described by a Hilbert space $\hilb$, whose energy is associated with a self-adjoint operator $H(\gamma)$ on this space---the Hamiltonian of the system---in the form
\begin{equation}\label{eq:hamiltonian}
    H(\gamma)=\gamma H_0+V,
\end{equation} 
with $H_0$ representing the unperturbed Hamiltonian of the system, $V$ being a perturbation, and $\gamma>0$ controls the strength of the unperturbed part. In order to accommodate situations in which the spectrum of $H_0$ admits continuous components, we shall take $\hilb$ as a possibly infinite-dimensional Hilbert space. For now, we shall restrict ourselves to the situation in which both operators $H_0$ and $V$ are bounded, $H_0,V\in\mathcal{B}(\hilb)$, postponing the extension of our results to unbounded $H_0$ to Sec.~\ref{sec:unbounded}\@.

Let $\sigma(H_0)$ be the spectrum of $H_0$. We shall take a decomposition of $\sigma(H_0)$ into finitely many disjoint components,
\begin{equation}\label{eq:coarse}
    \sigma(H_0)=\bigcup_{k=1}^m\sigma_k, \qquad \sigma_k\cap\sigma_l=\emptyset\quad(k\neq l),
\end{equation}
with the quantity
\begin{equation}
    \eta=\min_{k\neq l}\dist(\sigma_k,\sigma_l)>0
    \label{eq:SpectralGap}
\end{equation}
representing the \textit{spectral gap} of the unperturbed Hamiltonian $H_0$. This induces a decomposition of the Hilbert space of the system into mutually orthogonal components,
\begin{equation}\label{eq:decomposition}
    \hilb=\bigoplus_{k=1}^m\hilb_k,\qquad \hilb_k=P_k\hilb,
\end{equation}
where $P_k$ is the orthonormal projection corresponding to the $k$th component of the spectrum. The set of projections $\{P_k\}_{k=1,\dots,m}$ form a complete orthonormal family of projections on $\hilb$, i.e.
\begin{equation}
    P_k=P_k^\dag,\quad P_kP_l=P_k\delta_{kl},\quad \sum_kP_k=\openone,
\end{equation}
which commute with $H_0$,
\begin{equation}
    [P_k,H_0]=0,
\end{equation}
or equivalently, such that $H_0$ is block-diagonal with respect to the decomposition~\eqref{eq:decomposition} of the Hilbert space of the system. As such, so is the evolution group $\rme^{-\ii tH_0}$ generated by $H_0$: any vector $\psi\in\hilb_k$ for some $k=1,\dots,m$ will only evolve inside the space $\hilb_k$. It must be remarked that as each $P_k$ projects onto a Hilbert space with dimension possibly higher than one---in fact, possibly infinite---the decomposition~\eqref{eq:decomposition} of $\hilb$ could be interpreted as a \textit{coarse-grained} decomposition of the Hilbert space. Figure \ref{fig:theFigure} illustrates a potential coarse-graining of a Hilbert space based on the spectrum of a generic Hamiltonian.

When considering a perturbed Hamiltonian $H(\gamma)=\gamma H_0+V$ with a generic Hermitian perturbation $V$, the block-diagonal structure of the unperturbed Hamiltonian $H_0$ is generally spoiled---\textit{leakage} between distinct spaces $\hilb_k$ will generally occur. The leakage out of the $k$th eigenspace at time $t$ can be quantified by
\begin{equation}\label{eq:leakage}
    \mathcal{L}_k(t;\gamma)=\|Q_k\e^{-\ii tH(\gamma)}P_k\|,
\end{equation}
where $Q_k=\sum_{l\neq k}P_l=\openone-P_k$ is the projection onto the orthogonal complement of the $k$th component of the Hilbert space, and $\| X\|$ is the operator norm of $X$\@.

Following Refs.~\cite{Soliverez1981,eternal}, our goal will be to construct an effective adiabatic generator $H_\mathrm{eff}(\gamma)$ satisfying the following three requirements.
\begin{enumerate}
    \item \label{req1}
    $H_\mathrm{eff}(\gamma)$ must be \textit{block-diagonal} with respect to the projections $P_k$, that is, $[P_k,H_\mathrm{eff}(\gamma)]=0$ for all $k$, thus causing no leakage between separate spectral subspaces;
    \item  \label{req2}
    $H_\mathrm{eff}(\gamma)$ should be \textit{similar} (and thus \textit{isospectral}) to the original Hamiltonian $H(\gamma)$, that is, $H_\mathrm{eff}(\gamma)=\Omega^{-1}(\gamma) H(\gamma)\Omega(\gamma)$ for some bounded invertible transformation $\Omega(\gamma)$ with bounded inverse;
    \item  \label{req3}
    The similarity transformation $\Omega(\gamma)$ should be close to the identity, $\|\Omega(\gamma)-\openone\| =\mathcal{O}(\gamma^{-1})$, for large $\gamma$.
\end{enumerate}
If such an operator $H_\mathrm{eff}(\gamma)$ is found, the leakage out of the $k$th component of the Hilbert space can be equivalently quantified as
\begin{equation}\label{eq:leakage_bis}
    \mathcal{L}_k(t;\gamma)=\|Q_k(\e^{-\ii tH(\gamma)}-\e^{-\ii tH_\mathrm{eff}(\gamma)})P_k\|,
\end{equation}
as $H_\mathrm{eff}(\gamma)$, being block-diagonal, satisfies $Q_k\e^{-\ii tH_\mathrm{eff}(\gamma)}P_k=0$. Since the operator norm is submultiplicative, this quantity is bounded above by 
\begin{align}\label{eq:bounding_the_leakage}
    \mathcal{L}_k(t;\gamma)&\leq\|Q_k\|\|\e^{-\ii tH(\gamma)}-\e^{-\ii tH_\mathrm{eff}(\gamma)}\|\|P_k\|\nonumber\\
    &=\|\e^{-\ii tH(\gamma)}-\e^{-\ii tH_\mathrm{eff}(\gamma)}\|,
\end{align}
since $\|P_k\|=\|Q_k\|=1$. Therefore, the distance between the evolutions generated by the actual Hamiltonian $H$ and any effective generator $H_\mathrm{eff}(\gamma)$ serves as a uniform bound for the leakage out of \emph{any} subspace $\mathcal{H}_k=P_k\mathcal{H}$.

To avoid a cumbersome notation, we will hereafter make the $\gamma$-dependence of operators and scalars implicit, and simply write $H(\gamma)\equiv H$, $\mathcal{L}_k(t;\gamma) \equiv \mathcal{L}_k(t)$, and so on. Table~\ref{tab:gammadependence} serves as a reference to the implicit $\gamma$-dependence of all objects introduced in this manuscript. 
\begin{table}[ht]
\caption{List of all relevant operators and scalar quantities introduced throughout the text, according to whether they depend on the value of $\gamma$ or not.}
\centering
    \begin{tabular}{|c|c|}
        \hline 
        \textbf{Independent of $\gamma$} & \textbf{Implicitly $\gamma$-dependent} \\
        \hline
        $H_0$, $V$, $P_k$, $Q_k$ & $H$, $H_{\text{eff}}$, $H_{\rm Bloch}$, $H_{\rm SW}$ \\
        $t$, $\eta$ & $\mathcal{D}_{\rm Bloch}$, $\mathcal{D}_{\rm SW}$, $\mathcal{L}_k(t)$, $\delta$, $\varepsilon$ \\
        $\Omega^{(j)}$, $\Omega^{(j)}_k$, $Y^{(j)}$, $Y^{(j)}_k$  &  $\Omega$, $\Omega_k$, $W$, $W_k$   \\
        $H_0^{(n)}$, $\eta_{n}$ & $H^{(n)}$, $H_{\rm Bloch}^{(n)}$, $\varepsilon_{n}$ \\
        \hline
    \end{tabular}  \label{tab:gammadependence}
\end{table}

In the following sections, we will explicitly provide two effective choices for the generator $H_\mathrm{eff}$, one obtained through an adaptation of the Bloch method to coarse-grained spectral decompositions (Sec.~\ref{sec:bloch}), and one derived from the former through a Schrieffer--Wolff transformation (Sec.~\ref{sec:sw}). This will also enable us to find a time-independent bound on the leakage $\mathcal{L}_k(t)$, which can be even extended to the case in which $H_0$ is unbounded (Sec.~\ref{sec:unbounded}).

\section{Effective Dynamics via the Bloch Method}\label{sec:bloch}
\subsection{Adiabatic Bloch Generator}
In order to construct an effective generator $H_\mathrm{eff}$ having the desired properties, we will pursue the following strategy. We will search for a family of bounded operators $\Omega_1,\dots,\Omega_m\in\mathcal{B}(\hilb)$ satisfying the \emph{Bloch equations}~\cite{durand_direct_1983},
\begin{gather}\label{eq:bloch_eq_1}
H\Omega_k=\Omega_kH\Omega_k, \\\label{eq:bloch_eq_2}
\Omega_kP_k=\Omega_k,\\\label{eq:bloch_eq_3}
P_k\Omega_k=P_k.
\end{gather}
Assuming that such a solution is found, one defines the \textit{adiabatic Bloch generator} $H_\mathrm{Bloch}$ via
\begin{equation}
    H_\mathrm{Bloch}=\sum_kP_kH\Omega_k,
\end{equation}
and the \textit{Bloch wave operator} $\Omega$ via
\begin{equation}
    \Omega=\sum_k\Omega_k.
\end{equation}
The adiabatic Bloch generator $H_\mathrm{Bloch}$ constructed in this way satisfies the following properties:
\begin{proposition}\label{prop:diagonalizes}
    Let $\{\Omega_k\}_{k=1,\dots,m}\subset\mathcal{B}(\hilb)$ be a solution of the Bloch equations~\eqref{eq:bloch_eq_1}--\eqref{eq:bloch_eq_3}. Then, the adiabatic Bloch generator $H_\mathrm{Bloch}$ satisfies
    \begin{equation}\label{eq:h_to_hbloch}
        H\Omega=\Omega H_\mathrm{Bloch},
    \end{equation}
    and it is block-diagonal with respect to the projections $\{P_k\}_{k=1,\dots,m}$,
    \begin{equation}\label{eq:bloch_is_block}
        H_\mathrm{Bloch}P_k=P_kH_\mathrm{Bloch}.
    \end{equation}
\end{proposition}
\begin{proof}
    For every $k$, one has
    \begin{equation}\label{eq:bloch_eq_4}
        P_kH_\mathrm{Bloch}=\sum_{l}P_kP_lH\Omega_l=P_kH\Omega_k,
    \end{equation}
    where we used the orthogonality of the projections $P_k$ and $P_l$ for $k\neq l$. Therefore,
    \begin{equation} \label{eq:bloch_eq18}
        H\Omega_k=\Omega_kH\Omega_k=\Omega_kP_kH\Omega_k=\Omega_kH_\mathrm{Bloch},
    \end{equation}
    where we used Eqs.~\eqref{eq:bloch_eq_1},~\eqref{eq:bloch_eq_2}, and~\eqref{eq:bloch_eq_4}. Summing this equation over $k$ yields Eq.~\eqref{eq:h_to_hbloch}.

    Furthermore, using Eq.~\eqref{eq:bloch_eq_2}, one has
    \begin{equation}\label{eq:bloch_eq_5}
        H_\mathrm{Bloch}P_k=\sum_lP_lH\Omega_lP_k=P_kH\Omega_k.
    \end{equation}
    Equations~\eqref{eq:bloch_eq_4} and~\eqref{eq:bloch_eq_5} imply the commutativity~\eqref{eq:bloch_is_block}.
\end{proof}

In general, the system of equations~\eqref{eq:bloch_eq_1}--\eqref{eq:bloch_eq_3} may admit multiple solutions. We will use this freedom to our advantage, and construct a solution which satisfies requirements \ref{req1}--\ref{req3} listed in Sec.~\ref{sec:setting} to be regarded as a good effective generator. By Proposition~\ref{prop:diagonalizes}, requirement~\ref{req1} is automatically satisfied by any generator constructed in this way. For the remaining two requirements, we need to ensure the following properties, at least for sufficiently large $\gamma$:
\begin{itemize}
    \item The Bloch operator $\Omega$ must be invertible, so that the relation $H\Omega=\Omega H_{\rm Bloch}$ implies $H_{\rm Bloch}=\Omega^{-1}H\Omega$, and thus similarity;
    \item The Bloch operator must satisfy
\begin{equation}
    \|\Omega-\openone\|=\mathcal{O}(\gamma^{-1}).
\end{equation}
\end{itemize}

\subsection{Perturbative Solution of the Bloch Equations}
In order to find a solution to the Bloch equations~\eqref{eq:bloch_eq_1}--\eqref{eq:bloch_eq_3} with the desired properties, we will use a perturbative Ansatz for the solutions---namely, for every $k=1,\dots,m$, we will search for a solution in the form
\begin{equation}\label{eq:ansatz}
    \Omega_k=\sum_{j=0}^\infty\frac{1}{\gamma^j}\Omega_k^{(j)}
\end{equation}
for some sequence of operators $\{\Omega_k^{(j)}\}_{j\in\mathbb{N}}\subset\mathcal{B}(\hilb)$.
Of course, once such a solution is found, the convergence of this series must be checked a posteriori.

Plugging the Ansatz~\eqref{eq:ansatz} into Eq.~\eqref{eq:bloch_eq_1} together with the definition $H=\gamma H_0+V$, a straightforward order-by-order substitution shows that the Bloch equations are reduced in the following form: for every $k=1,\dots,m$, Eq.~\eqref{eq:bloch_eq_1} becomes
\begin{equation} \label{eq:bloch-commutator-equ0}
\begin{cases}
\medskip
\displaystyle
[H_0,\Omega_k^{(0)}]=0&(j=0),\\
\displaystyle
[H_0,\Omega_k^{(j)}]=-V\Omega_k^{(j-1)}+\sum_{i=0}^{j-1}\Omega_k^{(i)}V\Omega_k^{(j-1-i)}&(j \geq 1).
\end{cases}
\end{equation}
Equations~\eqref{eq:bloch_eq_2} and~\eqref{eq:bloch_eq_3} are rephrased as
\begin{equation}\label{eq:bloch_eq_2_bis}
\Omega_k^{(j)}P_k=\Omega_k^{(j)}\quad(j\geq 0)
\end{equation}
and
\begin{equation}\label{eq:bloch_eq_3_bis}
\begin{cases}
\medskip
\displaystyle
P_k\Omega_k^{(0)}=P_k&
(j=0),\\
\displaystyle
P_k\Omega_k^{(j)}=0 \ &
(j\geq 1),
\end{cases}
\end{equation}
respectively.

Taking a look at the system~\eqref{eq:bloch-commutator-equ0}, we notice that, for every $j\in\mathbb{N}$, the $j$th equation depends on the solutions of the previous $j-1$ equations. As such, the solutions can be found recursively. To this extent, we notice that each of the equations in Eq.~\eqref{eq:bloch-commutator-equ0} is in the form
\begin{equation}\label{eq:sylvester}
    AX-XB=Y,
\end{equation}
with $A$, $B$, and $Y$ being bounded operators on $\hilb$. This is known as a \textit{Sylvester equation} in the literature. This equation can be solved explicitly: see Ref.~\cite[Theorem~VII.2.5]{bhatia_matrix_2013} and Lemma~\ref{lemma:sylvester} in Appendix~\ref{sec:app0}\@. As such, we can find an explicit recursive expression for the perturbative solution of the Bloch equations~\eqref{eq:bloch-commutator-equ0}--\eqref{eq:bloch_eq_3_bis}.
\begin{proposition}\label{prop:perturbative}
Let the bounded operators $\Omega_k^{(j)}$, for $k=1,\dots,m$ and $j\in\mathbb{N}$, be given by
\begin{align}
\label{eq:solution_12}
\begin{cases}
    \medskip
    \displaystyle
    \Omega_k^{(0)}=P_k&(j=0),\\
    \displaystyle
    \Omega_k^{(j)}
    =\int_{-\infty}^\infty\rme^{-\rmi tH_0}Q_k Y_k^{(j)}P_k\rme^{\rmi tH_0}f(t)\,\rmd t&(j\geq 1),
\end{cases}
\end{align}
where
\begin{equation}
Y_k^{(j)}
=
-V\Omega_k^{(j-1)}+\sum_{i=1}^{j-1}\Omega_k^{(i)}V\Omega_k^{(j-1-i)}\quad(j\geq 1),
\end{equation}
and $f(t)$ is any function in $L^1(\mathbb{R})$ such that 
\begin{equation}\label{eq:f(t)}
\hat{f}(s)=\int_{-\infty}^\infty f(t)\rme^{-\rmi st}\,\rmd t=\frac{1}{s}\quad\text{for}\quad |s|\ge\eta,
\end{equation}
with $\eta$ being the spectral gap of $H_0$ defined in Eq.~\eqref{eq:SpectralGap}. Then, $\Omega_k^{(j)}$ solve the Bloch equations~\eqref{eq:bloch-commutator-equ0}--\eqref{eq:bloch_eq_3_bis}.
\end{proposition}
We refer to Appendix~\ref{sec:app1} for the proof. This gives us a closed, easily computable expression for the series expansion of the operators $\Omega_k$ in Eq.~\eqref{eq:ansatz}, and therefore, for the Bloch wave operator $\Omega=\sum_k\Omega_k$, which we wanted to find. As discussed previously, we expect the series to converge and their sum to satisfy the desired property $\|\Omega-\openone\|\ll1$ for $\gamma\gg1$. This is indeed the case.
\begin{proposition}\label{prop:close_to_id}
    Assume that the parameter $\gamma$ in the Hamiltonian $H=\gamma H_0+V$ satisfies
    \begin{equation}
        \gamma>4\pi\frac{\|V\|}{\eta},
    \end{equation}
    with $\eta$ being the spectral gap~\eqref{eq:SpectralGap} of $H_0$. Then, the series~\eqref{eq:ansatz}, whose terms are the solutions~\eqref{eq:solution_12} of the Bloch equations,
    converges to a bounded operator $\Omega_k$ for all $k=1,\dots,m$.
    Moreover, the Bloch wave operator $\Omega=\sum_{k=1}^m \Omega_k$  satisfies
    \begin{equation}\label{eq:close_to_id}
        \|\Omega- \emph{\openone}\|\leq \delta \!\left(\frac{\|V\|}{\gamma\eta}\right) <1,
    \end{equation}
    where 
    \begin{equation}
    \label{eq:delta}
        \delta(x)=\frac{(1-\sqrt{1-4\pi x})^2}{4\pi x}.
    \end{equation}
\end{proposition}
We refer to Appendix~\ref{sec:app2} for the proof. 
The bound on $\|\Omega-\openone\|$ can be directly and explicitly computed. In particular, since $\delta(x)\sim \pi x$ as $x\to0$,
one has $\delta=\mathcal{O}(\gamma^{-1})$ as $\gamma\to\infty$, as expected.

Remarkably, this bound only depends on the unperturbed Hamiltonian $H_0$ by the spectral gap $\eta$: neither the number $m$ of coarse-grained components of the spectrum (bands) nor their dimensionality enters the bound. This makes it extremely robust against modifications of the physical setup under examination.

\subsection{The Bloch Effective Generator}\label{sec:bound}
At this stage, we provided an explicit construction of a solution $\{\Omega_k\}_{k=1,\dots,m}$ of the Bloch equations whose corresponding Bloch wave operator $\Omega=\sum_k\Omega_k$ is close to the identity by virtue of Eq.~\eqref{eq:close_to_id}. We already know from Proposition~\ref{prop:diagonalizes} that the corresponding adiabatic Bloch generator $H_\mathrm{Bloch}=\sum_kP_kH\Omega_k$ is block-diagonal with respect to the spectral decomposition of $H_0$. We can now conclude our reasoning and show that indeed $H_\mathrm{Bloch}$ induces an adiabatic evolution \textit{eternally} close to the one generated by $H$.
\begin{thm}\label{thm:main}
Let $H=\gamma H_0+V$, with the parameter $\gamma$ satisfying
   \begin{equation}
        \gamma>4\pi\frac{\|V\|}{\eta},
    \end{equation}
    where $\eta$ is the spectral gap~\eqref{eq:SpectralGap} of $H_0$.
    Let $\Omega_k\in\mathcal{B}(\hilb)$, for ${k=1,\dots,m}$, be the sums of the convergent series~\eqref{eq:ansatz} whose terms are given by Eq.~\eqref{eq:solution_12}. 
    Then, the (non-Hermitian) operator 
    \begin{equation}
        H_\mathrm{Bloch}=\sum_kP_kH\Omega_k
    \end{equation}
    satisfies the following properties:
    \begin{itemize}
    \item[(i)] $H_\mathrm{Bloch}$ is similar to $H$ as 
    \begin{equation}
    H_\mathrm{Bloch}=\Omega^{-1}H\Omega,
    \end{equation}
    where $\Omega=\sum_k\Omega_k$;
    \item[(ii)] $H_\mathrm{Bloch}$ is block-diagonal with respect to $H_0$, namely,
    \begin{equation}
    [H_\mathrm{Bloch}, P_k]=0,
    \end{equation}
    for all $k$;
    \item [(iii)] for all $t\in\mathbb{R}$, the distance $\mathcal{D}_{\rm Bloch}$ between the dynamics generated by $H$ and $H_{\rm Bloch}$ reads
    \begin{equation}
      \mathcal{D}_{\rm Bloch} \equiv \|\e^{-\ii tH}-\e^{-\ii tH_\mathrm{Bloch}}\| 
      \leq \varepsilon\!\left(\frac{\|V\|}{\gamma \eta}\right),
    \label{eq:ourBound1} 
    \end{equation}
where
\begin{equation}
\varepsilon(x)  =\frac{1}{\sqrt{1-4\pi 
    x}}-1 .
    \label{eq:ourBound1epsilon} 
    \end{equation}
    \end{itemize}
\end{thm}
\begin{proof}
Properties (i) and (ii) hold by virtue of Proposition~\ref{prop:diagonalizes}, so we are left with proving (iii).
By Proposition~\ref{prop:close_to_id}, we know $\|\Omega-\openone\|\leq \delta\bigl(\frac{ \|V\|}{\gamma \eta}\bigr)<1$. 
    Following standard arguments (cf.~Appendix~\ref{sec:app1}), one can check that, as long as $\delta<1$, $\Omega$ is invertible, and the following three inequalities hold,
    \begin{gather}
        \|\Omega\|\leq 1+\delta,\label{eq:ineq_omega_1}\\
        \|\Omega^{-1}\|\leq\frac{1}{1-\delta},\label{eq:ineq_omega_2}\\
        \|\Omega^{-1}-\openone\|\leq\frac{\delta}{1-\delta}.\label{eq:ineq_omega_3}
    \end{gather}
    Besides, since 
    $H\Omega=\Omega H_\mathrm{Bloch}$, we have $H_\mathrm{Bloch}=\Omega^{-1} H\Omega$, and therefore,
    \begin{equation}\label{eq:evols}
        \e^{-\ii tH_\mathrm{Bloch}}
        =\Omega^{-1}\e^{-\ii tH}\Omega.
    \end{equation}
    Using the inequalities~\eqref{eq:ineq_omega_1}--\eqref{eq:ineq_omega_3} and the similarity~\eqref{eq:evols}, we are finally able to estimate $\mathcal{D}_{\rm Bloch}$ as follows:
\begin{align}
\label{eq:ourBound}
&\| \rme^{-\ii tH}-\rme^{-\ii tH_{\rmBloch}} \| \nonumber\\
&\quad= \| \rme^{-\ii tH}-\Omega^{-1}\rme^{-\ii tH}\Omega\|  
\nonumber\\
&\quad= \|{-(\Omega^{-1}-\openone)\rme^{-\ii tH}}-\Omega^{-1}\rme^{-\ii tH}(\Omega-\openone)\|  
\nonumber\\
&\quad\le\|\Omega^{-1}-\openone\|\|\rme^{-\ii tH}\|+\|\Omega^{-1}\|\|\rme^{-\ii tH}\|\|\Omega-\openone\|  
\nonumber\\
&\quad\le\frac{2\delta}{1-\delta},  
\end{align}
which by Eq.~\eqref{eq:delta} gives the claimed bound.    
\end{proof}
As anticipated, point (iii) also directly enables us to bound above the leakage $\mathcal{L}_k(t)$ out of any given spectral component of the system through Eq.~\eqref{eq:bounding_the_leakage}. We postpone a detailed analysis of this bound to Sec.~\ref{sec:unbounded}, where its extension to systems with unbounded energy spectrum is also discussed.

\section{Effective Dynamics Generated by a Schrieffer--Wolff Hamiltonian}\label{sec:sw}
A shortcoming of the Bloch construction presented in Sec.~\ref{sec:bloch} is, of course, the fact that $H_\mathrm{Bloch}$ is generally non-Hermitian, and thus induces a nonunitary evolution. We will show how to fix this issue through the Schrieffer--Wolff construction~\cite{schrieffer_relation_1966}, which allows us to transform $H_\mathrm{Bloch}$ into a Hermitian Hamiltonian $H_\mathrm{SW}$ still satisfying all required properties---it is block-diagonal and, as $\gamma\to\infty$, its evolution converges to the real one.

We begin by recalling that, for $\gamma> 4 \pi\|V\|/\eta$, the Bloch wave operator $\Omega$
of Theorem~\ref{thm:main} is invertible. Consequently, the operator $\Omega^\dag\Omega$ is also invertible and, by construction, nonnegative, so that one can define a new transformation
\begin{equation} \label{eq:W-def}
    W= \Omega(\Omega^\dagger  \Omega)^{-1/2}.
\end{equation}
Then, $W$ can be shown to be the unitary Schrieffer--Wolff transformation on the system~\cite[Definition~3.1]{bravyi_schriefferwolff_2011}, that is, $W=\sum_kW_k$, where $W_k$ is the direct rotation from the unperturbed spectral subspace $\hilb_k$ of $H_0$ with orthonormal projection $P_k$ to the corresponding perturbed one $\tilde{\hilb}_k$ of $H$.
A proof based on Ref.~\cite[Sec.~VIII]{eternal} is provided in Appendix~\ref{sec:swproof}\@.

Consequently, we can show that the resulting effective Hamiltonian is both Hermitian and block-diagonal:
\begin{proposition}\label{prop:sw_effective}
The Schrieffer--Wolff effective generator, defined through
\begin{equation}\label{eq:sw}
    H_\mathrm{SW} = W^\dagger H W,
\end{equation}
is Hermitian and satisfies $[H_\mathrm{SW},P_k]=0$.
\end{proposition}
\begin{proof}
    $H_\mathrm{SW}$ is clearly Hermitian as it is unitarily equivalent to the Hermitian operator $H$. The fact that it is block-diagonal follows from the block-diagonality of $\Omega^\dagger \Omega$, proven in Eq.~\eqref{eq:omegadaggeromegadiagonal} in Appendix~\ref{sec:swproof},
    \begin{align}
        H_\mathrm{SW} &= [(\Omega^\dagger  \Omega)^{-1/2}]^\dagger \Omega^\dagger H \Omega(\Omega^\dagger  \Omega)^{-1/2} \nonumber \\
        &= (\Omega^\dagger  \Omega)^{-1/2}\Omega^\dagger \Omega H_\mathrm{Bloch} (\Omega^\dagger  \Omega)^{-1/2} \nonumber \\
        &=(\Omega^\dagger  \Omega)^{1/2}H_\mathrm{Bloch} (\Omega^\dagger  \Omega)^{-1/2},
    \end{align}
    so $H_\mathrm{SW}$ is the product of three diagonal operators. For the second equality, we used Proposition~\ref{prop:diagonalizes}\@.
\end{proof}

Given the new unitary transformation $W$ and the Hermitian diagonal generator $H_\mathrm{SW}$, with similar methods as before, we can bound the distance between the Schrieffer--Wolff transformation and the identity, and between the effective evolution and the original one.  
\begin{thm}\label{thm:sw-distance}
    Let $H=\gamma H_0+V$, and let $W$ be the Schrieffer--Wolff unitary given by Eq.~\eqref{eq:W-def}. Assume that the parameter $\gamma$ satisfies 
   \begin{equation} \label{eq:gammacriteriaSW}
        \gamma>\frac{2\pi}{\sqrt{2}-1}\frac{\|V\|}{\eta},
    \end{equation}
    with $\eta$ being the spectral gap~\eqref{eq:SpectralGap} of $H_0$.
    Then, the Hermitian operator 
    \begin{equation}
        H_\mathrm{SW}= W^\dagger H W
    \end{equation}
    satisfies the following properties:

    \begin{itemize}
    \item[(i)] $H_\mathrm{SW}$ is block-diagonal with respect to $H_0$, namely,
    \begin{equation}
    [H_\mathrm{SW},P_k]=0,
    \end{equation}
    for all $k$;
        \item [(ii)] $W$ is close to the identity, i.e.,
        \begin{equation}\label{eq:theoremwminus1}
            \|W-\openone\| \leq \frac{1+\delta}{\sqrt{1-2\delta-\delta^2}}-1,
        \end{equation} 
        with $\delta= \delta\bigl(\frac{ \|V\|}{\gamma \eta}\bigr)$ as per Eq.~\eqref{eq:delta};
    \item [(iii)] the distance $\mathcal{D}_{\rm SW}$ between the dynamics generated by $H$ and $H_{\rm SW}$ is bounded uniformly for all $t\in\mathbb{R}$ by
        \begin{align} 
          \mathcal{D}_{\rm SW} &\equiv \|\e^{-\ii tH}-\e^{-\ii tH_\mathrm{SW}}\| \nonumber \\
          &\leq2\left( \frac{1}{\sqrt{\sqrt{1-4\pi \frac{\|V\|}{\gamma \eta}}-2\pi\frac{\|V\|}{\gamma \eta}}}-1\right). \label{eq:DSW}  
    \end{align}
    \end{itemize}
\end{thm}
\begin{proof}
    Property (i) is true through Proposition~\ref{prop:sw_effective}. From the inequality~\eqref{eq:gammacriteriaSW}, one can prove by simple algebra the inequality $\delta< \sqrt{2}-1$: this is done in Appendix~\ref{sec:appSWbounds}, where we also show that, as a consequence, the following inequalities hold:
\begin{gather}
    \|\Omega^\dag\Omega-\openone\| 
    \le2\delta+\delta^2, 
    \label{eq:omegadaggeromegabound1}\\
    \|(\Omega^\dag\Omega)^{-1/2}\|
    \le(1-2\delta-\delta^2)^{-1/2},
    \label{eq:omegadaggeromegabound2}\\
    \|(\Omega^\dag\Omega)^{-1/2}-\openone\|
    \le(1-2\delta-\delta^2)^{-1/2}-1.
    \label{eq:omegadaggeromegabound3}
\end{gather}
    Using the definition~\eqref{eq:W-def}, we bound the distance in (ii) as
    \begin{align}\label{eq:Wminus1}
        \|W-\openone\| 
        &= \|\Omega(\Omega^\dagger \Omega)^{-1/2}-\openone\|  \nonumber \\
        &= \|(\Omega-\openone)(\Omega^\dagger \Omega)^{-1/2}+[(\Omega^\dagger \Omega)^{-1/2}-\openone]\|  \nonumber \\
        &\leq \|\Omega-\openone\| \|(\Omega^\dagger \Omega)^{-1/2}\|+ \|(\Omega^\dagger \Omega)^{-1/2}-\openone\| \nonumber \\
        &
        \leq \frac{1+\delta}{\sqrt{1-2\delta-\delta^2}}-1,
    \end{align}
    which proves the inequality~\eqref{eq:theoremwminus1}. The distance in (iii) is bounded as
    \begin{align}
        &\|\e^{-\ii tH}-\e^{-\ii tH_\mathrm{SW}}\| \nonumber \\
        &\quad= \|\e^{-\ii tH}-W^{\dagger}e^{-\ii tH}W\| \nonumber \\
        &\quad= \| {(W-\openone)W^\dag\e^{-\ii tH}} -W^{\dagger}e^{-\ii tH}(W-\openone) \| \nonumber \\
        &\quad\leq 2 \|W-\openone\| \nonumber \\
        &\quad
        \leq 2\left(\frac{1+\delta}{\sqrt{1-2\delta-\delta^2}}-1\right)\nonumber \\
          &\quad=2\left( \frac{1}{\sqrt{\sqrt{1-4\pi \frac{\|V\|}{\gamma \eta}}-2\pi\frac{\|V\|}{\gamma \eta}}}-1\right),
    \end{align}
    which completes the proof.
\end{proof}
We note that, by construction, the bound~\eqref{eq:DSW} on the distance between the effective Schrieffer--Wolff evolution and the actual evolution is larger than the bound~\eqref{eq:ourBound1} on the Bloch evolution from the previous section. As such, we will utilize the latter to bound the leakage.

\section{The Bound on the Leakage}\label{sec:unbounded}
We can now proceed to find a universal bound on the leakage $\mathcal{L}_k(t)$ out of any spectral component of the unperturbed system. As discussed in Sec.~\ref{sec:setting} [cf.~Eq.~\eqref{eq:bounding_the_leakage}] and applying Theorem~\ref{thm:main}(iii), for any couple of Hermitian $H_0,V\in\mathcal{B}(\hilb)$ with $\gamma>4\pi \frac{\|V\|}{\eta}$, we have
\begin{equation}\label{eq:Lbound}
    \mathcal{L}_k(t)\leq\|\e^{-\ii tH}-\e^{-\ii tH_{\rm Bloch}}\|
    \leq \frac{1}{\sqrt{1-4\pi \frac{\|V\|}{\gamma \eta}}}-1.
\end{equation}
This shows that, for sufficiently large $\gamma$, the leakage is eternally bounded from above by a function of $\gamma$, the operator norm $\|V\|$ of the perturbation, and the spectral gap $\eta$ of $H_0$.

Importantly, the bound above depends on $H_0$ only by the value of its spectral gap $\eta$: neither the number of spectral bands nor their dimensionality plays any role. As such, it sounds reasonable that the bound can be directly extended to the case in which $H_0$ is a possibly \textit{unbounded} self-adjoint operator on $\hilb$, as long as the perturbation $V$ remains bounded. Physically, this covers the cases in which either the energy spectrum of the unperturbed Hamiltonian has an infinite number of spectral bands and/or has spectral bands of infinite size. To cover these cases, one might replicate the discussion in the previous section with $H_0$ unbounded---at the additional price of paying some attention to possible mathematical issues due to the unboundedness of the operators involved. Here, we will follow a more direct route, which will essentially consist of applying Theorem~\ref{thm:main} to a bounded approximation of the operators involved, and then taking the limit in a suitable sense.

We shall consider again a Hamiltonian on $\hilb$ in the form $H=\gamma H_0+V$, where $H_0$ is now a possibly unbounded self-adjoint operator representing the unperturbed Hamiltonian of the system, $V=V^\dag\in\mathcal{B}(\hilb)$ a bounded perturbation, and $\gamma>0$ as before. Since $V$ is bounded, $\gamma H_0+V$ is also a self-adjoint operator for every positive $\gamma$ on the domain of $H_0$. We shall again take a coarse-grained decomposition of the energy spectrum $\sigma(H_0)$ as in Eq.~\eqref{eq:coarse}, with $\{P_k\}_k$ being the corresponding projections---the only difference being that the number $m$ of coarse-grained components of the spectrum, and their corresponding projections, can be infinite. We shall again require that the (coarse-grained) spectral gap $\eta$ be strictly positive.
\begin{thm}\label{thm:main_unbounded}
    Let $H_0$ be a (possibly unbounded) self-adjoint operator on $\hilb$ and consider a spectral coarse-graining
    \begin{equation}\label{eq:coarse2}
    \sigma(H_0)=\bigcup_{k=1}^N\sigma_k, \qquad \sigma_k\cap\sigma_l=\emptyset\quad(k\neq l),
\end{equation} 
    with $N\in\mathbb{N}\cup\{+\infty\}$ and with positive spectral gap, 
    \begin{equation}
      \eta=\inf_{k\neq l}\dist(\sigma_k,\sigma_l)>0. 
    \end{equation}
    Let $V=V^\dag\in\mathcal{B}(\hilb)$, and $H=\gamma H_0+V$. Assume that the parameter $\gamma$ satisfies
    \begin{equation}
        \gamma>4\pi\frac{\|V\|}{\eta}.
    \end{equation}
    
  Then, for every bounded spectral component $\sigma_k$, the leakage out of $\sigma_k$, $\mathcal{L}_k(t)=\|Q_k\e^{-\ii tH}P_k\|$, satisfies the following inequality for all $t\in\mathbb{R}$:
    \begin{align}
    \mathcal{L}_k(t)
    &\leq \frac{1}{\sqrt{1-4\pi \frac{\|V\|}{\gamma \eta}}}-1 \\
    &\leq 9 \pi \frac{\|V\|}{\gamma \eta}.
    \end{align}
\end{thm}
\begin{proof}
    If $H_0$ is bounded, this directly comes from Theorem~\ref{thm:main}(iii). In the general case, we know that, by the spectral theorem for self-adjoint operators, we can define a truncated Hamiltonian $H_0^{(n)}$ by
\begin{equation}\label{eq:truncated}
    H_0^{(n)}= H_0  P_{H_0}([-n,n]) + E_nP_{H_0}(\mathbb{R}\setminus[-n,n]),
\end{equation}
where $P_{H_0}(I)$ is the spectral projection of $H_0$ onto the subspace of vectors with energy belonging to the interval $I$, and $E_n$ is an arbitrary energy in $\sigma(H_0)\cap [-n,n]$.

By construction, the truncated Hamiltonian~\eqref{eq:truncated} is a bounded operator for every $n\in\mathbb{N}$, whose spectrum is a finite truncation of the unbounded spectrum of $H_0$:  
\begin{equation}
 \sigma(H_0^{(n)})=\sigma(H_0)\cap[-n,n].   
\end{equation}
The role of the second term in the right-hand side of Eq.~\eqref{eq:truncated} is to avoid adding the zero eigenvalue if it is not present in the spectrum of $H_0$.

In Appendix~\ref{sec:app4}\@, we prove that, in the limit $n\to\infty$, the dynamics generated by $H^{(n)}:=\gamma H_0^{(n)}+V$ reproduces the one of $H$, in the following sense,
\begin{equation}\label{eq:approx}
        \lim_{n\to\infty}\|(\e^{-\ii tH^{(n)}}-\e^{-\ii tH})\psi\|=0,
\end{equation}
for any $\psi\in\hilb$.
Consider the leakage out of the $k$th coarse-grained component of $\sigma(H_0)$, with $P_k=P_{H_0}(\sigma_k)$ being the corresponding projection. Since $\sigma_k$ is bounded, one can take $n$ large enough so that $\sigma_k$ is completely contained in $[-n,n]$. Then, by suitably choosing $E_n$ in Eq.~\eqref{eq:truncated}, $P_k$ can be made to coincide with $P_{H_0^{(n)}}(\sigma_k)$, the spectral projection of $H_0^{(n)}$ corresponding to $\sigma_k$, whence $[H_\mathrm{Bloch}^{(n)},P_k]=0$, and therefore $Q_k\e^{-\ii tH_\mathrm{Bloch}^{(n)}}P_k=0$. Consequently, for every $\psi\in\hilb$,
\begin{align}\label{eq:leakage_unbounded}
    \|Q_k\e^{-\ii tH}P_k\psi\|={}&\|Q_k(\e^{-\ii tH}-\e^{-\ii tH_\mathrm{Bloch}^{(n)}})P_k\psi\|\nonumber\\
    \leq{}&\|Q_k(\e^{-\ii tH}-\e^{-\ii tH^{(n)}})P_k\psi\|\nonumber\\
    &{}+\|Q_k(\e^{-\ii tH^{(n)}}-\e^{-\ii tH_\mathrm{Bloch}^{(n)}})P_k\psi\|.
\end{align}
The first term goes to zero as $n\to\infty$ by Eq.~\eqref{eq:approx}, while the second one can be estimated via Theorem~\ref{thm:main}\@. To apply the latter, we need to be sure that, under our assumptions, the following condition holds:
\begin{equation}
        \gamma>\frac{4\pi\|V\|}{\eta_{n}},
\end{equation}
where $\eta_{n}$ is the spectral gap of $H_0^{(n)}$.
To this extent, note that necessarily $\eta\leq\eta_{n}$ (when truncating the spectrum, we can only increase the value of the spectral gap), so that indeed $\gamma>4\pi\|V\|/\eta\geq4\pi\|V\|/\eta_{n}$. 
In addition, define [cf.~Eq.~\eqref{eq:ourBound1epsilon}]
\begin{equation}
    \varepsilon_{n}=\varepsilon\!\left(\frac{\|V\|}{\gamma\eta_{n}}\right),
    \end{equation}
then, since $\varepsilon(x)$ in Eq.~\eqref{eq:ourBound1epsilon} is an increasing function of $x$, we also have $\varepsilon_{n}\leq\varepsilon_{\infty}=\varepsilon\bigl(\frac{ \|V\|}{\gamma \eta}\bigr)$.
Consequently,
    \begin{align}\label{eq:leakage_n}
        &\|Q_k(\e^{-\ii tH^{(n)}}-\e^{-\ii tH_\mathrm{Bloch}^{(n)}})P_k\psi\|\nonumber\\
        &\quad\leq\|Q_k\|\|\e^{-\ii tH^{(n)}}-\e^{-\ii tH_\mathrm{Bloch}^{(n)}}\|\|P_k\psi\|\nonumber\\
    &\quad\leq\varepsilon_n\|\psi\|\nonumber\\
        &\quad\leq \varepsilon_{\infty}\|\psi\|.
    \end{align}
Taking the limit $n\to\infty$ on both sides of Eq.~\eqref{eq:leakage_unbounded} and using Eq.~\eqref{eq:leakage_n}, we get
    \begin{equation}
         \|Q_k\e^{-\ii tH}P_k\psi\|\leq\varepsilon_{\infty}\|\psi\|,
    \end{equation}
    whence
    \begin{align}
        \mathcal{L}_k(t)&=\|Q_k\e^{-\ii tH}P_k\|
        \nonumber\\
         & \leq\varepsilon\!\left(\frac{\|V\|}{\gamma\eta}\right) \nonumber\\
         &= \frac{1}{\sqrt{1-4\pi \frac{\|V\|}{\gamma \eta}}}-1 \nonumber\\
    &\leq 9 \pi \frac{\|V\|}{\gamma \eta}.
    \end{align}
Notice that the first bound on the leakage exceeds the maximal meaningful value of $2$ whenever the scaling $\frac{\|V\|}{\gamma \eta}$ is larger than $2/(9 \pi)$. Therefore, the linear upper bound  $9 \pi \frac{\|V\|}{\gamma \eta}$ holds for all $\gamma>0$.
\end{proof}

\section{Examples}\label{sec:examples}

\subsection{The Linear Chain of Unit Cells}\label{ch:chain}
First, we consider a 1D chain of $n$ unit cells in the tight-binding model with 3 sites per unit cell with periodic boundary conditions. The Hamiltonian of the system is
\begin{align}
   H_0= \sum_{j=0}^{n-1}\,\Bigl(
   g_1 \ket{3j}\bra{3j+1} &{}+g_2 \ket{3j+1}\bra{3j+2} \nonumber\\
   &{}+g_3 \ket{3j+2}\bra{3j+3} +\mathrm{h.c.} \Bigr),
\end{align}
with $\ket{n}=\ket{0}$.
The model is usually solved using the method developed by Felix Bloch~\cite{Bloch1928}: one introduces the momentum basis $\ket{k,\alpha} \equiv \frac{1}{\sqrt{n}}\sum_{j=0}^{n-1} \e^{\ii k j}\ket{3j+\alpha}$, where $k$ is the quasi-momentum in the first Brillouin zone and $\alpha \in \{1,2,3\}$. The momentum-space Hamiltonian matrix can be represented by its matrix elements $H_{\alpha,\beta}(k) = \bra{k,\alpha} H_0 \ket{k,\beta}$:
\begin{equation}
    H(k) = \begin{pmatrix}
    0 & g_1 & g_3 \rme^{-\rmi k}\\
g_1 & 0 & g_2\\
g_3 \rme^{+\rmi k} & g_2 & 0
\end{pmatrix}.
\end{equation}
The three energy bands $E(k)$ are found by diagonalizing the matrix, whose associated characteristic polynomial reads
\begin{equation}
    E(k)^3 - E(k)(g_1^2+g_2^2+g_3^2)-2g_1 g_2 g_3 \cos k = 0.
\end{equation}

Let us consider the perturbation $V$ to be a local noise with relative strength $\epsilon$, such that 
\begin{equation}
    V= \epsilon \sum_i^{3n} r_i \ket{i}\bra{i} ,
\end{equation}
where $\{r_i\}$ are in the the interval $[-1,1]$.

For a chain with a perturbation of norm $\|V\| = 0.01$ and couplings $g_1=1$, $g_2=3/2$, and $g_3=2$, resulting in a bandgap $\eta \geq 1.15$, numerical simulations show that the leakage is indeed stable in time, besides being essentially independent of the size of the system $n$. Numerics for 50 to 500 unit cells result in a maximal leakage in the range of 0.004 and 0.007, depending on the parameter values. Yet numerical simulations cannot access arbitrarily long times or truly macroscopic chains, precisely the regime where our analytical bound becomes essential. For this example, our time-independent bound~\eqref{eq:Lbound} is 0.059, and holds for any chain length $n \in \mathbb{N}_{\ge 2} \cup \{\infty\}$.  
Moreover, as $\epsilon$ tends to 0, the simulated leakage scales as $\epsilon \sim 1/\gamma$, which agrees with the leading-order term of our bound that is likewise $\mathcal{O}(1/\gamma)$, showing that the bound captures the correct leading-order physics and is asymptotically tight. 

The 1D chain presented above is a simple example of a wider range of problems within the tight-binding model for which our methodology can deliver quantitative bounds. This applies to any crystal-like multidimensional structure with a finite bandgap. 
Hence, our result contributes to the robustness of the condensed-matter techniques themselves, ensuring that small enough perturbations only contribute to a finite leakage eternally.

\subsection{Harmonic Chain}\label{ch:harmonicchain}
An interesting example of a perturbed, unbounded system with an infinite number of bands is a one-dimensional array of harmonic oscillators.
This is used to model trapped ultracold atoms~\cite{endres_atom-by-atom_2016} and multimode photonic waveguides~\cite{kang_topological_2023}. 
Acting on the Hilbert space $\ell^{2}(\mathbb{N})^{\oplus n}$, the Hamiltonian is given by
\begin{align}
   H_0={}& \omega \sum_{i=1}^{n} \sum_{k=0}^{\infty}\left(k+\frac{1}{2} \right) \ket{k,i}\bra{k,i}    \nonumber \\
   &{}- g \sum_{i=1}^{n-1} \sum_{k=0}^{\infty}\,\Bigl(  \ket{k,i}\bra{k,i+1} + \ket{k,i+1}\bra{k,i}\Bigr),
\end{align}
with natural frequency $\omega$ and nearest-neighbor hopping amplitude $g$. Here, $\ket{k,i}$ designates the state with $k$ energy quanta in the $i$th oscillator sector and the vacuum in all other sectors.
In the weak-coupling regime corresponding to $\omega > 4 g$, the system admits separate Bloch-like energy bands of bandwidth $4g$, centered around the eigenvalues $\omega(k+1/2)$ of the harmonic ladder, resulting in a constant bandgap $\eta \geq \omega - 4g$, for any $n \in \mathbb{N}_{\ge 2} \cup \{\infty\}$.

Taking any bounded perturbation of this system $\|V\| \ll \eta$, we can apply our bound to the leakage out of any energy band of interest. As in the previous example, this perturbation could be either, similarly to the previous example, a local perturbation of the potential or hopping, or potentially next-to-nearest neighbor couplings. Here we consider the case where the perturbation  explicitly induces band transitions,
\begin{equation} \label{eq:harmonic_pert}
    V = \frac{v_0}{2} \sum_{i=1}^{n} \sum_{k=0}^{\infty} \,\Bigl(
    \ket{k,i}\bra{k+1,i} + \ket{k+1,i}\bra{k,i}
    \Bigr).
\end{equation}
This is a bounded perturbation with $\|V\|=v_0$. 
We obtain
\begin{equation}
    \mathcal{L}_k<\left(1-\frac{4\pi v_0}{\omega-4g}\right)^{-1/2}-1.
\end{equation}

Numerical simulations of this model again exhibit a leakage scaling proportional to $v_0 \sim 1/\gamma$ as $v_0$ tends to zero. The leading term in our analytic bound is likewise $\mathcal{O}(1/\gamma)$, reaffirming the asymptotic tightness of the bound across different illustrative examples.

\subsection{Transmon Architecture}\label{ch:transmon}
Current quantum-computing architectures often rely on an array of coupled transmons, each individually governed by the single Josephson-junction Hamiltonian~\cite{likharev_theory_1985}
\begin{equation}\label{eq:josephson}
    H_0 = -4 E_C \frac{\partial^2}{\partial \phi ^2} - E_J \cos\phi,
\end{equation}
where $\phi \in \mathbb{R}$ is the superconducting phase difference across the junction, $E_C $ the charge, 
and $E_J $ the Josephson energy
~\cite{koch_charge-insensitive_2007}. 
The time-independent Schr\"odinger equation associated with this Hamiltonian can be reduced to the well-known Mathieu ordinary differential equation. A recent extensive discussion of approximate solutions can be found in Ref.~\cite{wilkinson_approximate_2018}.

When $E_C \ll E_J$, the Josephson Hamiltonian~\eqref{eq:josephson} can be interpreted as the Hamiltonian for a tight-binding model with cosine potential. Following the arguments in Ref.~\cite{koch_charge-insensitive_2007}, it can be shown that this model admits energy bands
with an asymptotic expansion of the $k$th bandgap given by
\begin{align}
\frac{\eta_{k}}{E_C}
={}& 4\sqrt{\frac{E_J}{2E_C}}-1-k
\nonumber\\
& {}-\left(
     \frac{3}{32}
     + \frac{3}{32}(2k + 1)
     + \frac{3k^{2} + 3k + 1}{16}
\right)
\frac{1}{\sqrt{\frac{E_J}{2E_C}}}
\nonumber\\
&{} -\biggl(
     \frac{3}{256}
     + \frac{2k + 1}{16}
     + \frac{5}{128}(3k^{2} + 3k + 1)
\nonumber\\
&\qquad\qquad\quad {}+\frac{5}{256}(4k^{3} + 5k^{2} + 4k + 1)\biggl)\,
   \frac{1}{\frac{E_J}{2E_C}}
\nonumber\\
& {}+\mathcal{O}\biggl(\left(\frac{E_J}{2E_C}\right)^{-3/2}\biggr). \label{eq:josephson_bandgap}
\end{align}
This approximation becomes accurate when increasing 
$E_J/E_C$ but becomes worse for larger bandgap numbers $k$. For 
some experimentally relevant parameters~\cite{paik_observation_2011, barends_coherent_2013}, and when the subspaces of interest are only the lowest-energy eigenstates, this is a good approximation. 
We will therefore treat the model with a coarse-graining spectral decomposition consisting of three parts: the ground state, the first excited state, and all the rest of the spectrum, with bandgaps given by Eq.~(\ref{eq:josephson_bandgap}) 
and a minimal gap $\eta=\eta_1$.

As a perturbation, we consider a finite transparency of the tunnel barrier in a typical superconductor-insulator-superconductor design, characterized by the transmission coefficient $D$, where $0 < D \ll 1$. This phenomenon is studied extensively, for example, in the review article by Golubov \textit{et~al.}~\cite{golubov_current-phase_2004}\@.  
In this case, the Josephson Hamiltonian is of the form 
\begin{align}
H_J 
&= \Delta \sqrt{1 - D \sin^2(\phi/2)} 
\nonumber\\
&= \Delta \sqrt{1 - D/2}
  +\frac{\Delta D}{4\sqrt{1 - D/2}}\cos\phi +V( \phi),
\label{eq:HJ_expanded}
\end{align}
where $\Delta$ is the superconducting potential, and
\begin{align}
    V(\phi)={}& \Delta \sqrt{1 - D/2} \nonumber \\
    &{}\times\left( \sqrt{1+\frac{D/2}{1-D/2}\cos\phi} - 1 - \frac{D/4}{1-D/2}\cos\phi\right).
\end{align}
Up to an irrelevant constant, the cosine term in Eq.~\eqref{eq:HJ_expanded} combined with the charge energy constitutes the usual approximation~\eqref{eq:josephson} with $E_J=-\Delta\tfrac{D}{4\sqrt{1 - D/2}} $. The higher-order terms can be treated as a perturbation to the otherwise cosine potential, which is bounded by Lagrange's mean-value form. For $|\alpha|<1$, we have $\sqrt{1+\alpha} = 1+\frac{\alpha}{2}-\frac{\alpha^2}{8}(1+c)^{-3/2}$, for some real number $c$, satisfying $-|\alpha| \leq c \leq|\alpha| $. Adding $1$ to this inequality and raising each side to the negative power, we get $(1-|\alpha|)^{-3/2} \geq (1+c)^{-3/2} \geq (1+|\alpha|)^{-3/2}$. Applying this to $\alpha=  \frac{D/2}{1-D/2}\cos\phi$, the perturbation is bounded
\begin{align}
    \|V\| &\leq \frac{\Delta}{8} \sqrt{1 - D/2}  \left( \frac{D/2}{1-D/2}\right)^2  \left(1-\frac{D/2}{1-D/2} \right)^{-3/2}  \nonumber \\
    & \leq\Delta \frac{D^2}{32(1 - D/2)^{3/2}} \nonumber \\
    &=E_J\frac{D}{8(1 - D/2)}.
\end{align}
Our bounds then give 
\begin{equation}
    \mathcal{L}_k<\left(1-\frac{\pi D}{2-D}\frac{E_J}{E_C}\frac{1}{\eta/E_C}\right)^{-1/2}-1.
\end{equation}
Taking $E_J/E_C=90$ and a transparency of $D=10^{-3}$, we find that the leakage is less than $3 \times 10^{-3}$.

\section{Concluding Remarks}\label{sec:conclusion}
In this paper, we have established robust and time-independent---eternal---bounds on the leakage for quantum systems under the influence of perturbations. Importantly, they are directly formulated in a coarse-grained scenario: we can quantify the leakage out of the subspace of the Hilbert space corresponding to any subset of the spectrum of the unperturbed Hamiltonian---no matter whether the latter consists of one energy level, multiple energy levels, a continuous energy band, or any combination of these. The only relevant parameters that enter the bounds are (a) the coarse-grained spectral gap of the unperturbed system, and (b) the norm of the perturbation. Accordingly, our results even extend to quantum systems with unbounded energy spectra, as long as the perturbation remains bounded. In all cases, the leakage out of any subspace decays as $\gamma^{-1}$, with $\gamma$ the relative ratio between the unperturbed and perturbed components of the Hamiltonian. 

We achieved these bounds by constructing an effective propagator which does not mix subspaces corresponding to separate spectral components, and reproduces the target dynamics in the limit $\gamma\to\infty$. This was obtained by solving a system of operator equations, the Bloch equations. This method sometimes yields a generator, which is not self-adjoint (but still has a real spectrum). A self-adjoint generator can be obtained by means of the Schrieffer--Wolff transformation, for which we also bound the distance to the true evolution. Our analysis is completed by some examples which underscore the significance of our method for a wide range of quantum systems.

For future research, it would be interesting to extend our results to quantum systems subject to unbounded perturbations---a case which is currently not covered by our bounds, as they depend on the operator norm of the perturbation. Allowing the perturbation to be unbounded, one could expect the leakage to either decay more slowly than $\mathcal{O}(\gamma^{-1})$ or even to stay finite even for arbitrarily small perturbations, as the results in Ref.~\cite{facchi_robustness_2024} (for fine-grained leakage) suggest. This, on the other hand, is a physically meaningful problem, as many systems of interest from solid-state physics---whose energy spectra typically exhibit continuous bands, and which therefore represent natural test-beds for our coarse-grained approach---are indeed subject to unbounded perturbations in the thermodynamic limit. We expect our Bloch equation-based approach to be powerful enough to be extended to such cases, under suitable technical requirements on the interacting Hamiltonian; in such cases, we expect the bounds to be state-dependent. Lastly, it would be interesting and natural to adapt our approach to open quantum systems as well.

\begin{acknowledgments}
D.B. would like to acknowledge stimulating discussions with Philipp Hansmann. Z.S. was supported by the Sydney Quantum Academy. D.L. acknowledges financial support by Friedrich-Alexander-Universit\"at Erlangen-N\"urnberg through the funding program ``Emerging Talent Initiative'' (ETI), and was partially supported by the project TEC-2024/COM-84 QUITEMAD-CM\@.
P.F. acknowledges support from INFN through the project ``QUANTUM'', from the Italian National Group of Mathematical Physics (GNFM-INdAM), from PNRR MUR projects CN00000013-``Italian National Centre on HPC, Big Data and Quantum Computing'', and from the Italian funding within the ``Budget MUR - Dipartimenti di Eccellenza 2023--2027''  - Quantum Sensing
and Modelling for One-Health (QuaSiModO). 
K.Y. acknowledges supports by the Top Global University Project from the Ministry of Education, Culture, Sports, Science and Technology (MEXT), Japan, and by JSPS KAKENHI Grant No.~JP24K06904 from the Japan Society for the Promotion of Science (JSPS)\@.
\end{acknowledgments}

\appendix
\begin{widetext}
\section{Solving the Sylvester Equation}\label{sec:app0}
As remarked in the main text, the system of Bloch equations~\eqref{eq:bloch-commutator-equ0} that we need to solve involves a concatenated family of Sylvester equations, i.e.~operator equations in the form $AX-XB=Y$ for $A,B,Y\in\mathcal{B}(\hilb)$. For such equations, the following general result holds.
\begin{lemma}\label{lemma:sylvester}
    Let $A,B,Y\in\mathcal{B}(\hilb)$ be self-adjoint, and assume that the spectra $\sigma(A)$ and $\sigma(B)$ of $A$ and $B$ are disjoint, i.e.~$\eta=\dist(\sigma(A),\sigma(B))>0$. Then, the Sylvester equation $AX-XB=Y$ admits a bounded solution in the form
    \begin{equation}\label{eq:x}
        X=\int_{-\infty}^{\infty}\e^{-\ii tA}Y\e^{\ii tB}f(t)\,\mathrm{d}t,
    \end{equation}
    where $f$ is any function in $L^1(\mathbb{R})$ whose Fourier transform $\hat{f}$ satisfies
    \begin{equation}
        \hat{f}(s)=\int_{-\infty}^{\infty}f(t)\e^{-\ii st}\,\mathrm{d}t=\frac{1}{s}\quad\text{for}\quad|s|>\eta.
    \end{equation}
\end{lemma}
This result is claimed and proven in Ref.~\cite[Theorem~VII.2.5]{bhatia_matrix_2013}, by first tackling the case of finite-dimensional Hilbert spaces, for which the Sylvester equation is simply a matrix equation, and then by considering the infinite-dimensional limit. Here, we present a simpler proof directly adapted to the infinite-dimensional setting.
\begin{proof}
By virtue of the spectral theorem for self-adjoint operators on an infinite-dimensional Hilbert space, both operators $A$ and $B$ admit the following decompositions,
\begin{equation}
    A=\int_{\sigma(A)}a\,\mathrm{d}P_A(a),\qquad B=\int_{\sigma(B)}b\,\mathrm{d}P_B(b),
\end{equation}
where $P_A$ and $P_B$ are projection-valued measures on $\hilb$, and the integrals are to be understood in the Lebesgue sense. All standard manipulations for such integrals are allowed, provided that the noncommutativity of the two measures is taken into account. Besides, the measures are \textit{normalized} as
\begin{equation}
P_A(\sigma(A))=\int_{\sigma(A)}\mathrm{d}P_A(a)=\openone, \qquad
P_B(\sigma(B))=\int_{\sigma(B)}\mathrm{d}P_B(b)=\openone.
\end{equation}
Hence, plugging Eq.~\eqref{eq:x} in the operator $AX-XB$, one obtains
\begin{align}
    AX-XB &= \int_{\sigma(A)} \int_{\sigma(B)} \int_{-\infty}^\infty\rmd t\, f(t)\, \Bigl( \rmd P_A(a) A \rme^{-\rmi tA}Y\rme^{\rmi tB}\rmd P_B(b) - \rmd P_A(a)  \rme^{-\rmi tA}Y\rme^{\rmi tB} B\, \rmd P_B(b)   \Bigr)  \nonumber\\
    &= \int_{\sigma(A)} \int_{\sigma(B)} \int_{-\infty}^\infty\rmd t \,f(t) (a-b) \rme^{-\rmi t(a-b)} \rmd P_A(a) Y\,\rmd P_B(b)  \nonumber\\
    &=\int_{\sigma(A)} \int_{\sigma(B)} (a-b)\,\rmd P_A(a) Y\, \rmd P_B(b) \int_{-\infty}^\infty\rmd t\, f(t)\rme^{-\rmi t(a-b)}  \nonumber\\
    &=\int_{\sigma(A)} \int_{\sigma(B)}\rmd P_A(a) Y\, \rmd P_B(b)  \nonumber\\
    &= \openone Y \openone,
\end{align}
which concludes the proof.
\end{proof}

\section{Proof of Proposition~\ref{prop:perturbative}}\label{sec:app1}
We can now use Lemma~\ref{lemma:sylvester} to prove Proposition~\ref{prop:perturbative}, thus solving the system of Bloch equations~\eqref{eq:bloch-commutator-equ0}.
\begin{proof}[Proof of Proposition~\ref{prop:perturbative}]
We start with the zeroth order $j=0$. By $P_k\Omega_k^{(0)}=P_k$ and $\Omega_k^{(0)}P_k=\Omega_k^{(0)}$, we have
\begin{equation}
\Omega_k^{(0)}
=P_k\Omega_k^{(0)}P_k+Q_k\Omega_k^{(0)}P_k
=P_k+Q_k\Omega_k^{(0)}P_k.
\end{equation}
Moreover, as $[H_0,\Omega_k^{(0)}]=0$, we have $
[H_0,Q_k\Omega_k^{(0)}P_k]=0$ and then $Q_k\Omega_k^{(0)}P_k=0$. Therefore, $\Omega_k^{(0)}=P_k$.

We now analyze the higher-order contributions. For $j\ge1$, we have $P_k\Omega_k^{(j)}=0$ and $\Omega_k^{(j)}P_k=\Omega_k^{(j)}$, and hence,
\begin{equation}
\Omega_k^{(j)}=Q_k\Omega_k^{(j)}P_k.
\end{equation}
Inserting $\Omega_k^{(0)}=P_k$, we get
\begin{align}
[H_0,\Omega_k^{(j)}] &=-V\Omega_k^{(j-1)}+\sum_{i=0}^{j-1}\Omega_k^{(i)}V\Omega_k^{(j-1-i)} \nonumber
\\
&=-Q_kV\Omega_k^{(j-1)}+\sum_{i=1}^{j-1}\Omega_k^{(i)}V\Omega_k^{(j-1-i)} \nonumber
\\
&=Q_k Y_k^{(j)}P_k \quad(j\geq 1).
\end{align}
We can now solve the Sylvester equation
\begin{equation}
[H_0,Q_k\Omega_k^{(j)}P_k]=Q_kY_k^{(j)}P_k\quad(j\geq 1)
\end{equation}
by using Lemma~\ref{lemma:sylvester} with $A=H_0 Q_k$, $B=H_0 P_k$, $X=\Omega_k^{(j)}$, and $Y=Q_k Y_k^{(j)}P_k$  as
\begin{equation}
\Omega_k^{(j)}
=Q_k\Omega_k^{(j)}P_k
=\int_{-\infty}^\infty\rmd t\,\rme^{-\rmi tH_0}Q_k Y_k^{(j)}P_k\rme^{\rmi tH_0}f(t)\quad(j\geq 1), 
\end{equation}
where $f(t)$ is any function in $L^1(\mathbb{R})$ such that 
\begin{equation}
\hat{f}(s)=\int_{-\infty}^\infty\rmd t\,f(t)\rme^{-\rmi st}=\frac{1}{s}\quad\text{for}\quad |s|\ge\eta,
\end{equation}
with $\eta$ the spectral gap of $H_0$ defined in Eq.~\eqref{eq:SpectralGap}.
This completes the proof.
\end{proof}

\section{Proof of Proposition~\ref{prop:close_to_id}}\label{sec:app2}
To begin with, we define
\begin{equation}
\Omega^{(j)} \equiv \sum_k\Omega_k^{(j)},
\end{equation}
for every $j=0,1,2,\dots$, and
\begin{equation}\label{eq:omega_omegaj}
\Omega=\sum_{j=0}^\infty\frac{1}{\gamma^j}\Omega^{(j)}.
\end{equation}
Analogously,
\begin{align}
     Y^{(j)} &= \sum_k Y_k^{(j)}   \nonumber\\
    &=
    - \sum_k V \Omega_k^{(j-1)}+ \sum_k\sum_{i=1}^{j-1}\Omega_k^{(i)}V\Omega_k^{(j-1-i)}  \nonumber\\
    &=
    -  V \sum_k \Omega_k^{(j-1)}+ \sum_k\sum_{i=1}^{j-1}\Omega^{(i)}P_kV\Omega^{(j-1-i)} P_k  \nonumber\\
    &=
    -V\Omega^{(j-1)}+\sum_{i=1}^{j-1}\Omega^{(i)}\sum_kP_kV\Omega^{(j-1-i)}P_k\quad(j\geq 1).
\end{align}
Note that $\Omega_k^{(j)} = \Omega^{(j)} P_k$ and $Y_k^{(j)} = Y^{(j)} P_k$. With these definitions, we can rephrase the perturbative solution of the Bloch equations~\eqref{eq:bloch-commutator-equ0} provided in Proposition 
\ref{prop:perturbative} as
\begin{equation} \label{eq:omega-j-sol-2}
\Omega^{(0)}=\openone,\qquad\Omega^{(j)}
=\mathop{\sum\sum}_{k\neq\ell}
\int_{-\infty}^\infty\rmd t\,\rme^{-\rmi tH_0}P_k Y^{(j)}P_\ell\rme^{\rmi tH_0}f(t)\quad(j\ge1),
\end{equation}
again with $f\in L^1(\mathbb{R})$ satisfying the condition~\eqref{eq:f(t)}.
\begin{lemma}\label{lem:omegaj}
    For every $j=0,1,2,\dots$, the following estimate holds,
    \begin{equation} \label{eq:omega-j-bound-by-catalan}
\|\Omega^{(j)}\|
\le\left(\frac{\pi}{\eta}\|V\|\right)^jC_j\quad(j\geq 0),
\end{equation}
where $C_j$ is the $j$th Catalan number,
\begin{equation}
C_j=\frac{1}{j+1}
\begin{pmatrix}
2j\\
j
\end{pmatrix}
\quad(j\geq 0).
\end{equation}
\end{lemma}
\begin{proof}
We begin by noticing that the Catalan numbers satisfy the recurrence relation,
\begin{equation}\label{eq:catalan_recurrence}
C_0=1,\qquad
C_{j+1}=\sum_{i=0}^jC_iC_{j-i}\quad(j\geq 0).
\end{equation}
Now, $Y^{(1)}$ is bounded by
\begin{equation}
\|Y^{(1)}\|
=\|V\|,
\end{equation}
and, for $j\ge2$,
\begin{align}
\|Y^{(j)}\|
&\le
\|V\|\|\Omega^{(j-1)}\|
+\sum_{i=1}^{j-1}\|\Omega^{(i)}\|
\left\|
\sum_kP_kV\Omega^{(j-1-i)}P_k
\right\|
\nonumber\\
&=
\|V\|\|\Omega^{(j-1)}\|
+\sum_{i=1}^{j-1}\|\Omega^{(i)}\|
\sup_k\|P_kV\Omega^{(j-1-i)}P_k\|
\nonumber\\
&\le
\|V\|\|\Omega^{(j-1)}\|
+\sum_{i=1}^{j-1}\|\Omega^{(i)}\|
\|V\|\|\Omega^{(j-1-i)}\|
\nonumber\\
&=
\|V\|\left(
\|\Omega^{(j-1)}\|
+\sum_{i=1}^{j-1}\|\Omega^{(i)}\|
\|\Omega^{(j-1-i)}\|
\right).
\end{align}
In addition,
\begin{align}
\left\|
\mathop{\sum\sum}_{k\neq\ell}P_kYP_\ell
\right\|
&=
\left\|
Y-\sum_kP_kYP_k
\right\|
\nonumber\\
&\le
\|Y\|
+\left\|
\sum_kP_kYP_k
\right\|
\nonumber\\
&=
\|Y\|
+\sup_k\|P_kYP_k\|
\nonumber\\
&\le
\|Y\|
+\|Y\|
\nonumber\\
&=
2\|Y\|.
\vphantom{\biggl\|}
\end{align}
We can now bound $\Omega^{(j)}$ as follows. For $j=0$,
\begin{equation}
\|\Omega^{(0)}\|=1.
\end{equation}
For $j\geq1$, we use Eq.~\eqref{eq:omega-j-sol-2} and the property (cf.~Ref.~\cite[Theorem~VII.2.15]{bhatia_matrix_2013}) 
\begin{equation}
\inf_{f\in \mathcal{F}_\eta}\int_{\mathbb{R}}\rmd t\,|f(t)| =\frac{\pi}{2\eta}\,, \qquad 
\mathcal{F}_\eta=\left\{f\in L^1(\mathbb{R})\,:\,  \hat{f}(s)=1/s, \,\text{for }|s|\geq \eta\right\}.
\end{equation}
Then, we obtain
\begin{align} \label{eq:omega-j-bound}
\|\Omega^{(j)}\|
&
\le
\inf_{f\in \mathcal{F}_\eta}\int_{-\infty}^\infty\rmd t
\left\|
\mathop{\sum\sum}_{k\neq\ell}
P_k Y^{(j)}P_\ell
\right\|
|f(t)|
\nonumber\\
&
\le
\inf_{f\in \mathcal{F}_\eta} 2\|Y^{(j)}\|
\int_{-\infty}^\infty\rmd t\,
|f(t)|
\nonumber\\
&
=
\frac{\pi}{\eta}\|Y^{(j)}\|
\nonumber\\
&
\le
\frac{\pi}{\eta}
\|V\|\left(
\|\Omega^{(j-1)}\|
+\sum_{i=1}^{j-1}\|\Omega^{(i)}\|
\|\Omega^{(j-1-i)}\|
\right)
\nonumber\\
&
=
\frac{\pi}{\eta}
\|V\|
\sum_{i=0}^{j-1}\|\Omega^{(i)}\|
\|\Omega^{(j-1-i)}\|
\quad(j\geq 1).
\end{align}
We now prove the bound~\eqref{eq:omega-j-bound-by-catalan} by induction. The zeroth order is true, as $\|\Omega^{(0)}\|=1$. Assume that the bound holds up to the $j$th order. Then, by Eq.~\eqref{eq:omega-j-bound},
    \begin{align}
        \|\Omega^{(j+1)}\|
        &\le \frac{\pi}{\eta} \|V\|
        \sum_{i=0}^{j}\|\Omega^{(i)}\| \|\Omega^{(j-i)}\| \nonumber \\
        &\le\frac{\pi}{\eta} \|V\| \sum_{i=0}^{j} 
        \left(\frac{\pi}{\eta}\|V\|\right)^iC_i
        \left(\frac{\pi}{\eta}\|V\|\right)^{j-i}C_{j-i}\nonumber \\
        &=\left(\frac{\pi}{\eta}\|V\|\right)^{j+1}\sum_{i=0}^{j}C_i C_{j-i}\nonumber\\
        &=\left(\frac{\pi}{\eta}\|V\|\right)^{j+1}C_{j+1},
    \end{align}
    where we used Eq.~\eqref{eq:catalan_recurrence} in the last step, thus concluding the induction and the proof.
\end{proof}

We can finally prove Proposition~\ref{prop:close_to_id}\@.
\begin{proof}[Proof of Proposition~\ref{prop:close_to_id}] 
We first recall the generating function for the Catalan numbers~\cite{miana_catalan_2023}
\begin{equation}
G(x)=\sum_{j=0}^\infty C_jx^j=\frac{1-\sqrt{1-4x}}{2x}.
\end{equation}
By expanding $\Omega$ as in Eq.~\eqref{eq:omega_omegaj} and by using Lemma~\ref{lem:omegaj}, we get
\begin{align}
\|\Omega-\openone\|
&=\left\|
\sum_{j=1}^\infty\frac{1}{\gamma^j}\Omega^{(j)}
\right\|
\nonumber\\
&\le
\sum_{j=1}^\infty\frac{1}{\gamma^j}\|\Omega^{(j)}\|
\nonumber\\
&\le
\sum_{j=1}^\infty
\left(\frac{\pi}{\gamma\eta}\|V\|\right)^jC_j
\nonumber\\
&=
G\!\left(\frac{\pi}{\gamma\eta}\|V\|\right)-1
\nonumber\\
&=
\delta\!\left(\frac{\|V\|}{\gamma\eta}\right),
\end{align}
where 
\begin{equation}
        \delta(x)=G(\pi x)-1 =\frac{(1-\sqrt{1-4\pi x})^2}{4\pi x}
    \end{equation}
is the function appearing in the statement of Proposition~\ref{prop:close_to_id}\@.

The last thing that remains to be proven is that as long as $\gamma>4\pi \|V\|/\eta$, then $\|\Omega-\openone\|<1$. As $1\geq\sqrt{1-4\pi\frac{\|V\|}{\gamma \eta}}>0$, so $\sqrt{1-4\pi\frac{\|V\|}{\gamma \eta}} \geq 1-4\pi\frac{\|V\|}{\gamma \eta}$, and equivalently for any finite $\gamma$ it is true that $1-\sqrt{1-4\pi\frac{\|V\|}{\gamma \eta}} < 4\pi\frac{\|V\|}{\gamma \eta}$. Rearanging this inequality completes the proof $\|\Omega-\openone\|<1$.
\end{proof}

\section{Proof of Inequalities~\eqref{eq:ineq_omega_1}--\eqref{eq:ineq_omega_3}}\label{sec:app3}
We here derive the inequalities~\eqref{eq:ineq_omega_1}--\eqref{eq:ineq_omega_3}, which are used in the main text to prove Theorem~\ref{thm:main}\@. Let $\Omega\in\mathcal{B}(\hilb)$ such that $\|\Omega-\openone\|\leq\delta<1$. Then, the inequality~\eqref{eq:ineq_omega_1} follows immediately by the triangle inequality as
\begin{equation}
\|\Omega\| = \|\openone + \Omega - \openone\| \leq \|\openone\| + \|\Omega - \openone\| = 1 + \delta.
\end{equation}
We next prove the inequality~\eqref{eq:ineq_omega_2}. To begin with, recall that, given $T\in\mathcal{B}(\hilb)$ such that $\|T\|<1$, the operator $\openone-T$ is invertible, and its inverse is given by a Neumann series,
\begin{equation}
(\openone-T)^{-1}=\sum_{n=0}^\infty T^n.
\end{equation}
Replacing $T$ with $\openone-\Omega$, we are sure that the operator $\openone-(\openone-\Omega)=\Omega$ is invertible and
\begin{equation}
    \Omega^{-1}=\sum_{n=0}^\infty(\openone-\Omega)^n,
\end{equation}
whence, using the triangle inequality and the submultiplicativity of the operator norm,
\begin{equation}
    \|\Omega^{-1}\|\leq\sum_{n=0}^\infty\|(\openone-\Omega)^n\|\leq\sum_{n=0}^\infty\|\openone-\Omega\|^n\leq\sum_{n=0}^\infty\delta^n=\frac{1}{1-\delta},
\end{equation}
which proves the inequality~\eqref{eq:ineq_omega_2}. Finally, since
\begin{equation}
    \Omega^{-1}-\openone=\sum_{n=1}^\infty(\openone-\Omega)^n,
\end{equation}
we have
\begin{equation}
\|\Omega^{-1}-\openone\| 
\leq\sum_{n=1}^\infty\|(\openone-\Omega)^n\|\leq\sum_{n=1}^\infty\|\openone-\Omega\|^n\leq\sum_{n=1}^\infty\delta^n = \frac{\delta}{1-\delta},
\end{equation}
which proves the inequality~\eqref{eq:ineq_omega_3}.

\section{Schrieffer--Wolff Construction}
\label{sec:swproof}
\begin{proposition}\label{prop:sw_transformation}
Let $H=\gamma H_0+V$ with $\gamma> 4 \pi\|V\|/\eta$, and let $\Omega$ be the Bloch wave operator given in Theorem~\ref{thm:main}\@. Then,
\begin{equation} \label{eq:W-def2}
    W= \Omega(\Omega^\dagger  \Omega)^{-1/2}
\end{equation}
is the unitary Schrieffer--Wolff transformation on the system~\cite[Definition~3.1]{bravyi_schriefferwolff_2011}, that is, $W=\sum_kW_k$, where $W_k$ is the direct rotation from the unperturbed spectral subspace $\hilb_k$ of $H_0$ with orthonormal projection $P_k$ to the corresponding perturbed one $\tilde{\hilb}_k$ of $H$.
\end{proposition}

In order to prove this proposition, we shall first introduce a family of operators $\{\tilde{P}_k\}$ that serve as a perturbed counterpart of the projections $\{P_k\}$ in the presence of a nonzero coupling.
\begin{lemma}
For every $k=1,\dots,m$, the operators $\tilde{P}_k$ defined by 
\begin{equation}
    \tilde{P}_k = \Omega_k (\Omega^\dagger_k \Omega_k)^{-1} \Omega^\dagger_k, \label{eq:E1}
\end{equation}
where the inverse is understood on the subspace $\hilb_k$, satisfy the following properties,
\begin{gather}
    \tilde{P}_k^\dagger = \tilde{P}_k,  \label{eq:E2} \\
    \tilde{P}_k^2 =  \tilde{P}_k, \label{eq:E3} \\
    [H,  \tilde{P}_k] = 0,\label{eq:E4} \\ 
     \tilde{P}_k \rightarrow P_k\quad\mathrm{as}\quad\gamma \rightarrow \infty.\label{eq:E5}
\end{gather}
\end{lemma}
Therefore, $\tilde{P}_k$ are $\gamma$-dependent Hermitian projections with respect to which $H=\gamma H_0+V$ is block-diagonal, and they converge to the projections $P_k$ of the unperturbed Hamiltonian $H_0$ in the limit $\gamma\to\infty$. In simpler terms, for sufficiently large $\gamma$, 
the spectrum of the perturbed Hamiltonian $H$ is continuously deformed, and the operators $\tilde{P}_k$ are exactly the projections corresponding to the deformed spectral bands.
\begin{proof}
    The fact~\eqref{eq:E2} that all projections are Hermitian follows directly from their definition~\eqref{eq:E1}. In addition,
    \begin{equation}
        \tilde{P}_k^2 = \Omega_k (\Omega^\dagger_k \Omega_k)^{-1} (\Omega^\dagger_k\Omega_k) (\Omega^\dagger_k \Omega_k)^{-1} \Omega^\dagger_k = \Omega_k (\Omega^\dagger_k \Omega_k)^{-1} \Omega^\dagger_k = \tilde{P}_k.
    \end{equation}
    Equivalently to $H \Omega_k = \Omega_k H_\mathrm{Bloch}$ in Eq.~\eqref{eq:bloch_eq18}, by taking adjoints we get $\Omega_k^\dagger H = H_\mathrm{Bloch}^\dagger \Omega_k^\dagger $, so 
    \begin{equation}
        \Omega_k^\dagger \Omega_k H_\mathrm{Bloch} = \Omega_k^\dagger H \Omega_k = H_\mathrm{Bloch}^\dagger \Omega_k^\dagger \Omega_k,
    \end{equation}
    and 
    \begin{equation}
        H_\mathrm{Bloch} (\Omega_k^\dagger \Omega_k)^{-1} =  (\Omega_k^\dagger \Omega_k)^{-1} H_\mathrm{Bloch}^\dagger.
    \end{equation}
    Putting these equations together, we can prove the commutativity~\eqref{eq:E4},
    \begin{align}
        H \tilde{P}_k &= H \Omega_k (\Omega^\dagger_k \Omega_k)^{-1} \Omega^\dagger_k \nonumber \\
         &= \Omega_k H_\mathrm{Bloch} (\Omega_k^\dagger \Omega_k)^{-1}  \Omega^\dagger_k \nonumber \\
         &= \Omega_k (\Omega_k^\dagger \Omega_k)^{-1} H_\mathrm{Bloch}^\dagger  \Omega^\dagger_k \nonumber \\
         &=\Omega_k (\Omega^\dagger_k \Omega_k)^{-1} \Omega^\dagger_k H \nonumber \\
         &= \tilde{P}_k H.
    \end{align}
    Finally, the last equality follows from the fact that $\Omega_k$ approaches $P_k$ in the limit $\gamma \rightarrow \infty$.
\end{proof}

This lemma has two important consequences. First, it shows that the Bloch wave operator $\Omega_k$ links the perturbed projection $\tilde{P}_k$ with the unperturbed one $P_k$ via the equality $\tilde{P}_k \Omega_k = \Omega_k P_k = \Omega_k$. Secondly, it proves that, as long as the spectral gap remains finite, the operator $\Omega^\dagger \Omega$ is block-diagonal. This can be seen as follows: $P_k \Omega^\dagger \Omega P_l = \Omega^\dagger \tilde{P}_k\tilde{P}_l\Omega = 0 $, i.e. 
\begin{equation}\label{eq:omegadaggeromegadiagonal}
    \Omega^\dagger \Omega = \sum_k P_k \Omega^\dagger \Omega P_k.
\end{equation}

\begin{proof}[Proof of Proposition~\ref{prop:sw_transformation}]
The unitarity of the Schrieffer--Wolff transformation follows trivially from the definition~\eqref{eq:W-def} and the invertibility of $\Omega$,
\begin{align}
    W^\dagger W &=  [(\Omega^\dagger  \Omega)^{-1/2}]^\dagger \Omega^\dagger  \Omega(\Omega^\dagger  \Omega)^{-1/2} \nonumber \\
     &=  (\Omega^\dagger  \Omega)^{-1/2} (\Omega^\dagger  \Omega)^{1/2}(\Omega^\dagger  \Omega)^{1/2}(\Omega^\dagger  \Omega)^{-1/2} \nonumber \\
     &= \openone, \\
     W W^\dagger &= \Omega(\Omega^\dagger  \Omega)^{-1/2} [(\Omega^\dagger  \Omega)^{-1/2}]^\dagger \Omega^\dagger \nonumber \\
     &=\Omega (\Omega^\dagger  \Omega)^{-1} \Omega^\dagger \nonumber \\
     &= \Omega \Omega^{-1} (\Omega^\dagger) ^{-1} \Omega^\dagger \nonumber \\
     & = \openone.   
\end{align}
Following Ref.~\cite[Definitions~2.2 and~3.1]{bravyi_schriefferwolff_2011}, the Schrieffer--Wolff transformation of a subspace can be defined as the direct rotation from $P_k$ to $\tilde{P}_k$ by
\begin{equation}
    W_{k} =  \sqrt{\tilde{P}_k P_k}.
\end{equation}
We will show that this expression is equivalent to $\Omega_k (\Omega^\dagger_k \Omega_k)^{-1/2}$, where the inverse is understood on the subspace $P_k$. By squaring both expressions, we get
\begin{align}
    \left(\sqrt{\tilde{P}_k P_k}\right)^2 & = \Omega_k (\Omega^\dagger_k \Omega_k)^{-1} \Omega^\dagger_k P_k \nonumber \\
    & = \Omega_k (\Omega^\dagger_k \Omega_k)^{-1}, \\
    \Omega_k (\Omega^\dagger_k \Omega_k)^{-1/2} \Omega_k (\Omega^\dagger_k \Omega_k)^{-1/2}&= \Omega_k (\Omega^\dagger_k \Omega_k)^{-1/2} P_k \Omega_k (\Omega^\dagger_k \Omega_k)^{-1/2} \nonumber \\
    &= \Omega_k (\Omega^\dagger_k \Omega_k)^{-1/2} (\Omega^\dagger_k \Omega_k)^{-1/2} \nonumber \\
    &= \Omega_k (\Omega^\dagger_k \Omega_k)^{-1},
\end{align}
whence
\begin{align}
    \sum_k W_k &= \sum_k \Omega_k (\Omega^\dagger_k \Omega_k)^{-1/2}  \nonumber \\
    &= \sum_k \Omega P_k (P_k \Omega^\dagger \Omega P_k)^{-1/2} \nonumber \\
    &= \Omega (\Omega^\dagger \Omega)^{-1/2},
\end{align}
where the last equality follows from the block-diagonality~\eqref{eq:omegadaggeromegadiagonal}. This concludes the proof.
\end{proof}

\section{Proof of Inequalities~\eqref{eq:omegadaggeromegabound1}--\eqref{eq:omegadaggeromegabound3}} \label{sec:appSWbounds}
 Here, we prove the inequalities~\eqref{eq:omegadaggeromegabound1}--\eqref{eq:omegadaggeromegabound3}, necessary for proving Theorem~\ref{thm:sw-distance} in the main text. We begin by showing the equivalence of the inequalities $\gamma>\frac{2\pi}{\sqrt{2}-1}\frac{\|V\|}{\eta}$ and $\delta<\sqrt{2}-1$.
Introducing the scaling factor $\epsilon\equiv \frac{\pi \|V\|}{\gamma \eta} \geq 0$, we see 
\begin{align} \label{eq:F1}
    \gamma>\frac{2\pi}{\sqrt{2}-1}\frac{\|V\|}{\eta} &\iff \frac{\sqrt{2}-1}{2} > \epsilon \iff 
    0 > 2 \epsilon + (1-\sqrt{2}) \iff 0 > 4 \epsilon[2 \epsilon + (1-\sqrt{2})]  \nonumber \\
    &\iff 0> 8 \epsilon^2 + 4 \epsilon -4\sqrt{2}\epsilon \iff -4\epsilon > 8 \epsilon^2-4\sqrt{2}\epsilon  \iff
    1-4\epsilon > 1+8 \epsilon^2-4\sqrt{2}\epsilon=(1-2\sqrt{2}\epsilon)^2.
\end{align}
For any given $0<\epsilon<\frac{\sqrt{2}-1}{2}<\frac{1}{4}$, the last inequality in Eq.~\eqref{eq:F1} is equivalent to 
\begin{equation}
    \sqrt{1-4\epsilon}> 1-2\sqrt{2}\epsilon \iff \frac{1-\sqrt{1-4\epsilon}}{2 \epsilon} < \sqrt{2}.
\end{equation}
Recalling the definition~\eqref{eq:delta} of $\delta \equiv \frac{1-\sqrt{1-4\epsilon}}{2 \epsilon} -1$, we see that the two conditions on $\gamma$ and on $\delta$ are equivalent. Similar argument shows that $\epsilon<1/4$ implies $\delta \geq 0$.

To prove the inequality~\eqref{eq:omegadaggeromegabound1}, first recall that $\|\Omega - \openone\|= \|\Omega^\dagger - \openone\|<\delta$ and $\|\Omega\|=\|\Omega^\dagger\|<1+\delta$. By using the triangle inequality,
\begin{align} \label{eq:appF3}
     \|\Omega^\dagger \Omega - \openone\| &= \|\Omega^\dagger (\Omega-\openone)+(\Omega^\dagger- \openone)\| \nonumber \\
     &\leq  \|\Omega^\dagger \| \|\Omega-\openone\|+\|\Omega^\dagger- \openone\| \nonumber \\
     &\leq 2\delta+\delta^2.      
\end{align}
From the condition $0\leq \delta < \sqrt{2}-1$ follows $\|\Omega^\dagger \Omega - \openone\|<1$. Hence, $\Omega^\dagger \Omega=\openone +(\Omega^\dagger \Omega - \openone)$ must be invertible, the spectrum $\sigma(\Omega^\dagger \Omega)$ lies in the positive interval $[1-(2+\delta)\delta, 1+(2+\delta)\delta]$, and the inverse is given by a Neumann series, with the spectral norm of the square root satisfying 
\begin{align}
    \|(\Omega^\dag\Omega)^{-1/2}\|
    &=\|[\openone+(\Omega^\dag\Omega-\openone)]^{-1/2}\|
    \nonumber\\
    &\le\Bigl(1-\|\Omega^\dag\Omega-\openone\|\Bigr)^{-1/2}
    \nonumber\\
    &\le(1-2\delta-\delta^2)^{-1/2}.
\end{align}
Similarly,
\begin{align}
    \|(\Omega^\dag\Omega)^{-1/2}-\openone\|
    &=\|[\openone+(\Omega^\dag\Omega-\openone)]^{-1/2}-\openone\|
    \nonumber\\
    &\le\Bigl(1-\|\Omega^\dag\Omega-\openone\|\Bigr)^{-1/2}-1
    \nonumber\\
    &\le(1-2\delta-\delta^2)^{-1/2}-1.
\end{align}
which is the inequality~\eqref{eq:omegadaggeromegabound3} that was to be proven.

\section{Proof of Eq.~\eqref{eq:approx}}\label{sec:app4}
Here, we show that, for any $\psi\in\hilb$, the approximation property~\eqref{eq:approx} holds with $H=\gamma H_0+V$, $H^{(n)}=\gamma H_0^{(n)}+V$, and $H_0^{(n)}$ as given by Eq.~\eqref{eq:truncated}. Here, $H_0$ is an unbounded self-adjoint operator, densely defined on some domain $\operatorname{Dom} H_0\subset\hilb$; since $V=V^\dag$ is bounded, $H$ is also an unbounded self-adjoint operator with domain $\operatorname{Dom} H=\operatorname{Dom} H_0$. Instead, for every $n$, $H_0^{(n)}$ and $H^{(n)}$ are bounded operators defined on the whole Hilbert space $\hilb$. The following property holds by definition of $H_0^{(n)}$,
\begin{equation}
    \lim_{n\to\infty}\|H_0^{(n)}\psi-H_0\psi\|=0,\quad\text{for all}\ \psi\in\Dom H_0,
\end{equation}
and therefore
\begin{equation}\label{eq:approx_strong}
    \lim_{n\to\infty}\|H^{(n)}\psi-H\psi\|=0,\quad\text{for all}\ \psi\in\Dom H_0.
\end{equation}
Since $H$ is unbounded, Eq.~\eqref{eq:approx_strong} is generally \textit{not} sufficient by itself to prove Eq.~\eqref{eq:approx}. 

However, in our case, we can make use of the following sufficient condition for dynamical convergence (cf.~Ref.~\cite[Propositions~10.1.8 and~10.1.18]{deoliveira_intermediate_2009}.
\begin{proposition}\label{prop:sufficient_src}
    Let $(A_n)_{n\in\mathbb{N}}$ and $A$ be self-adjoint operators on $\hilb$. Assume that the following properties hold:
    \begin{enumerate}
        \item $\Dom A\subset\Dom A_n$ for all $n\in\mathbb{N}$;
        \item $A_n\psi\to A\psi$ for all $\psi\in\Dom A$.
        \end{enumerate}
        Then, $A_n$ converges to $A$ in the strong dynamical sense, i.e.,
        \begin{equation}
            \lim_{n\to\infty}\|(\e^{-\ii tA_n}-\e^{-\ii tA})\psi\|=0,\quad\text{for all}\ t\in\mathbb{R},\ \psi\in\hilb.
        \end{equation}
\end{proposition}
In our specific problem, it is clear that the truncated operators $H^{(n)}$ satisfy both assumptions of Proposition~\ref{prop:sufficient_src}\@. Indeed, $\Dom H^{(n)}=\hilb$ for all $n$, so that the first condition is obvious; the second condition coincides with Eq.~\eqref{eq:approx_strong}. Therefore, $H^{(n)}$ converges to $H$ in the strong 

dynamical sense, which is equivalent to Eq.~\eqref{eq:approx}.

\end{widetext}

\bibliography{CoarseGrainedBloch}

\end{document}